\documentclass[prx,aps,twoside,superscriptaddress]{revtex4}

\usepackage{amsmath,amsfonts,amssymb,caption,color,epsfig,graphics,graphicx,hyperref,latexsym,mathrsfs,revsymb,theorem,url,verbatim,enumerate,epstopdf,tikz,float,multirow,booktabs,appendix}
\hypersetup{colorlinks,linkcolor={blue},citecolor={red},urlcolor={blue}}
\usetikzlibrary{arrows, decorations.markings}

\tikzstyle{vecArrow} = [thick, decoration={markings,mark=at position
	1 with {\arrow[semithick]{open triangle 60}}},
double distance=1.4pt, shorten >= 5.5pt,
preaction = {decorate},
postaction = {draw,line width=1.4pt, white,shorten >= 4.5pt}]
\tikzstyle{innerWhite} = [semithick, white,line width=1.4pt, shorten >= 4.5pt]

\newtheorem{definition}{Definition}
\newtheorem{proposition}{Proposition}
\newtheorem{lemma}{Lemma}

\newtheorem{theorem}{Theorem}
\newtheorem{corollary}[definition]{Corollary}
\newtheorem{conjecture}[definition]{Conjecture}

\newtheorem{remark}[definition]{Remark}
\newtheorem{example}{Example}
\newtheorem{question}[definition]{Question}

\def\bcj{\begin{conjecture}}
	\def\ecj{\end{conjecture}}
\def\bcr{\begin{corollary}}
	\def\ecr{\end{corollary}}
\def\bd{\begin{definition}}
	\def\ed{\end{definition}}
\def\bea{\begin{eqnarray}}
	\def\eea{\end{eqnarray}}
\def\bem{\begin{enumerate}}
	\def\eem{\end{enumerate}}
\def\bex{\begin{example}}
	\def\eex{\end{example}}
\def\bim{\begin{itemize}}
	\def\eim{\end{itemize}}
\def\bl{\begin{lemma}}
	\def\el{\end{lemma}}
\def\bma{\begin{bmatrix}}
	\def\ema{\end{bmatrix}}
\def\bpf{\begin{proof}}
	\def\epf{\end{proof}}
\def\bpp{\begin{proposition}}
	\def\epp{\end{proposition}}
\def\bqu{\begin{question}}
	\def\equ{\end{question}}
\def\br{\begin{remark}}
	\def\er{\end{remark}}
\def\bt{\begin{theorem}}
	\def\et{\end{theorem}}

\def\squareforqed{\hbox{\rlap{$\sqcap$}$\sqcup$}}
\def\qed{\ifmmode\squareforqed\else{\unskip\nobreak\hfil
		\penalty50\hskip1em\null\nobreak\hfil\squareforqed
		\parfillskip=0pt\finalhyphendemerits=0\endgraf}\fi}
\def\endenv{\ifmmode\;\else{\unskip\nobreak\hfil
		\penalty50\hskip1em\null\nobreak\hfil\;
		\parfillskip=0pt\finalhyphendemerits=0\endgraf}\fi}
\newenvironment{proof}{\noindent \textbf{{Proof.~} }}{\qed}
\def\Dbar{\leavevmode\lower.6ex\hbox to 0pt
	{\hskip-.23ex\accent"16\hss}D}
\makeatletter
\def\url@leostyle{%
	\@ifundefined{selectfont}{\def\UrlFont{\sf}}{\def\UrlFont{\small\ttfamily}}}
\makeatother
\def\bcj{\begin{conjecture}}
	\def\ecj{\end{conjecture}}
\def\bcr{\begin{corollary}}
	\def\ecr{\end{corollary}}
\def\bd{\begin{definition}}
	\def\ed{\end{definition}}
\def\bea{\begin{eqnarray}}
	\def\eea{\end{eqnarray}}
\def\bem{\begin{enumerate}}
	\def\eem{\end{enumerate}}
\def\bex{\begin{example}}
	\def\eex{\end{example}}
\def\bim{\begin{itemize}}
	\def\eim{\end{itemize}}
\def\bl{\begin{lemma}}
	\def\el{\end{lemma}}
\def\bpf{\begin{proof}}
	\def\epf{\end{proof}}
\def\bpp{\begin{proposition}}
	\def\epp{\end{proposition}}
\def\bqu{\begin{question}}
	\def\equ{\end{question}}
\def\br{\begin{remark}}
	\def\er{\end{remark}}
\def\bt{\begin{theorem}}
	\def\et{\end{theorem}}

\def\btb{\begin{tabular}}
	\def\etb{\end{tabular}}

\newcommand{\nc}{\newcommand}

\nc{\bbA}{\mathbb{A}} \nc{\bbB}{\mathbb{B}} \nc{\bbC}{\mathbb{C}}
\nc{\bbD}{\mathbb{D}} \nc{\bbE}{\mathbb{E}} \nc{\bbF}{\mathbb{F}}
\nc{\bbG}{\mathbb{G}} \nc{\bbH}{\mathbb{H}} \nc{\bbI}{\mathbb{I}}
\nc{\bbJ}{\mathbb{J}} \nc{\bbK}{\mathbb{K}} \nc{\bbL}{\mathbb{L}}
\nc{\bbM}{\mathbb{M}} \nc{\bbN}{\mathbb{N}} \nc{\bbO}{\mathbb{O}}
\nc{\bbP}{\mathbb{P}} \nc{\bbQ}{\mathbb{Q}} \nc{\bbR}{\mathbb{R}}
\nc{\bbS}{\mathbb{S}} \nc{\bbT}{\mathbb{T}} \nc{\bbU}{\mathbb{U}}
\nc{\bbV}{\mathbb{V}} \nc{\bbW}{\mathbb{W}} \nc{\bbX}{\mathbb{X}}
\nc{\bbZ}{\mathbb{Z}}

\nc{\bA}{{\bf A}} \nc{\bB}{{\bf B}} \nc{\bC}{{\bf C}}
\nc{\bD}{{\bf D}} \nc{\bE}{{\bf E}} \nc{\bF}{{\bf F}}
\nc{\bG}{{\bf G}} \nc{\bH}{{\bf H}} \nc{\bI}{{\bf I}}
\nc{\bJ}{{\bf J}} \nc{\bK}{{\bf K}} \nc{\bL}{{\bf L}}
\nc{\bM}{{\bf M}} \nc{\bN}{{\bf N}} \nc{\bO}{{\bf O}}
\nc{\bP}{{\bf P}} \nc{\bQ}{{\bf Q}} \nc{\bR}{{\bf R}}
\nc{\bS}{{\bf S}} \nc{\bT}{{\bf T}} \nc{\bU}{{\bf U}}
\nc{\bV}{{\bf V}} \nc{\bW}{{\bf W}} \nc{\bX}{{\bf X}}
\nc{\bZ}{{\bf Z}} \nc{\bm}{{\bf m}} \nc{\bv}{{\bf v}}
\nc{\ba}{{\bf a}} \nc{\be}{{\bf e}} \nc{\bu}{{\bf u}}
\nc{\brr}{{\bf r}}

\nc{\cA}{{\cal A}} \nc{\cB}{{\cal B}} \nc{\cC}{{\cal C}}
\nc{\cD}{{\cal D}} \nc{\cE}{{\cal E}} \nc{\cF}{{\cal F}}
\nc{\cG}{{\cal G}} \nc{\cH}{{\cal H}} \nc{\cI}{{\cal I}}
\nc{\cJ}{{\cal J}} \nc{\cK}{{\cal K}} \nc{\cL}{{\cal L}}
\nc{\cM}{{\cal M}} \nc{\cN}{{\cal N}} \nc{\cO}{{\cal O}}
\nc{\cP}{{\cal P}} \nc{\cQ}{{\cal Q}} \nc{\cR}{{\cal R}}
\nc{\cS}{{\cal S}} \nc{\cT}{{\cal T}} \nc{\cU}{{\cal U}}
\nc{\cV}{{\cal V}} \nc{\cW}{{\cal W}} \nc{\cX}{{\cal X}}
\nc{\cZ}{{\cal Z}}

\nc{\hA}{{\hat{A}}} \nc{\hB}{{\hat{B}}} \nc{\hC}{{\hat{C}}}
\nc{\hD}{{\hat{D}}} \nc{\hE}{{\hat{E}}} \nc{\hF}{{\hat{F}}}
\nc{\hG}{{\hat{G}}} \nc{\hH}{{\hat{H}}} \nc{\hI}{{\hat{I}}}
\nc{\hJ}{{\hat{J}}} \nc{\hK}{{\hat{K}}} \nc{\hL}{{\hat{L}}}
\nc{\hM}{{\hat{M}}} \nc{\hN}{{\hat{N}}} \nc{\hO}{{\hat{O}}}
\nc{\hP}{{\hat{P}}} \nc{\hR}{{\hat{R}}} \nc{\hS}{{\hat{S}}}
\nc{\hT}{{\hat{T}}} \nc{\hU}{{\hat{U}}} \nc{\hV}{{\hat{V}}}
\nc{\hW}{{\hat{W}}} \nc{\hX}{{\hat{X}}} \nc{\hZ}{{\hat{Z}}}

\nc{\hn}{{\hat{n}}}

\newcommand{\ket}[1]{|#1\rangle}

\newcommand{\ketbra}[2]{|#1\rangle\!\langle#2|}

\def\Dbar{\leavevmode\lower.6ex\hbox to 0pt
	{\hskip-.23ex\accent"16\hss}D}

\begin{document}
	
	\title{Strong quantum nonlocality from hypercubes}
	
	
	\author{Fei Shi}
	\thanks{F. Shi and M.-S. Li contributed equally to this work.}
	\affiliation{School of Cyber Security,
		University of Science and Technology of China, Hefei, 230026, People's Republic of China}
	
	\author{Mao-Sheng Li}
	\thanks{F. Shi and M.-S. Li contributed equally to this work.}
	\affiliation{Department of Physics, Southern University of Science and Technology, Shenzhen 518055, China}
	\affiliation{Department of Physics, University of Science and Technology of China, Hefei 230026, China}
		
	\author{Mengyao Hu}
	\affiliation{School of Mathematical Sciences, Beihang University, Beijing 100191, China}
	
	\author{Lin Chen}
	\affiliation{LMIB and School of Mathematical Sciences, Beihang University, Beijing 100191, China}
	\affiliation{International Research Institute for Multidisciplinary Science, Beihang University, Beijing 100191, China}
	
	\author{Man-Hong Yung}
	\affiliation{Department of Physics, Southern University of Science and Technology, Shenzhen 518055, China}
	\affiliation{Institute for Quantum Science and Engineering, and Department of Physics, Southern University of Science and Technology, Shenzhen, 518055, China}

	\author{Yan-Ling Wang}
\email[]{Corresponding authors: drzhangx@ustc.edu.cn, wangylmath@yahoo.com}
\affiliation{School of Computer Science and Techonology, Dongguan University of Technology, Dongguan, 523808, China}	
	
	\author{Xiande Zhang}
	\email[]{Corresponding authors: drzhangx@ustc.edu.cn, wangylmath@yahoo.com}
	\affiliation{School of Mathematical Sciences,
		University of Science and Technology of China, Hefei, 230026, People's Republic of China}

	\begin{abstract}
	A set of multipartite orthogonal product states is strongly nonlocal if it is locally irreducible in every bipartition.  Most known constructions of strongly nonlocal orthogonal product set (OPS) are limited to tripartite systems, and they are lack of intuitive structures. In this work,  based on the decomposition for the outermost layer of an $n$-dimensional hypercube for $n=3,4,5$,  we successfully  construct strongly nonlocal  OPSs in any possible  three, four and five-partite systems, which answers an open question given by Halder \emph{et al.} [\href{https://journals.aps.org/prl/abstract/10.1103/PhysRevLett.122.040403}{Phys. Rev. Lett. \textbf{122}, 040403 (2019)}] and Yuan \emph{et al.} [\href{https://journals.aps.org/pra/abstract/10.1103/PhysRevA.102.042228}{Phys. Rev. A \textbf{102}, 042228 (2020)}] for any possible three, four and five-partite systems.  Our results build the connection between hypercubes and strongly nonlocal OPSs, and exhibit the phenomenon of strong quantum nonlocality without entanglement in multipartite systems. 	
	\end{abstract}

\maketitle
	
\vspace{-0.5cm}
~~~~~~~~~~\indent{\textbf{Keywords}}: strong quantum nonlocality, orthogonal product sets, hypercubes

\section{Introduction}

A set of multipartite orthogonal quantum states  is locally indistinguishable if it is not possible to optimally distinguish the states
by any sequence of local operations and classical communications (LOCC). When the classical
message is encoded in such multipartite states, it cannot be completely retrieved under LOCC. It requires global operations to retrieve the message.  Subsequently, local indistinguishability can be used for quantum data hiding \cite{terhal2001hiding,divincenzo2002quantum,eggeling2002hiding,Matthews2009Distinguishability} and quantum secret sharing \cite{Markham2008Graph,Hillery,Rahaman}.    Bennett \emph{et al.} provided a locally indistinguishable  orthogonal product basis in  $3\otimes 3$ \cite{bennett1999quantum},  which shows the phenomenon of quantum nonlocality without entanglement. Later, the locally indistinguishable orthogonal product states and orthogonal entangled  states have attracted much attention \cite{1,2,3,4,5,6,7,8,9,10,11,12,13,14,15,16,17}.

	Recently, Halder \emph{et al.} introduced the concept of locally irreducible set \cite{Halder2019Strong}. An OPS is locally irreducible  means that it is not possible to eliminate one or more states from the set by orthogonality-preserving local measurements. Local irreducibility ensures local indistinguishability, while the converse is not true usually.  An effective way to prove that an OPS is locally irreducible is to show that only trivial orthogonality-preserving local measurement can be performed to this set.    An OPS is strongly nonlocal if it is locally irreducible in every bipartition.  They also showed two strongly nonlocal orthogonal product bases in $3\otimes 3\otimes 3$ and $4\otimes 4\otimes 4$, respectively, which shows the phenomenon of strong quantum nonlocality without entanglement. After that,  based on the local irreducibility in some multipartitions, Zhang \emph{et al.} generalized the strong quantum nonlocality and gave some explicit examples  \cite{zhangstrong2019}. Yuan \emph{et al.} showed a strongly nonlocal OPS in $d\otimes d\otimes d$, $d\otimes d\otimes (d+1)$, $3\otimes 3\otimes 3\otimes 3$ and $4\otimes 4\otimes 4\otimes 4$  for $d\geq 3$ \cite{yuan2020strong}.   Further, Shi \emph{et al.} constructed a strongly nonlocal orthogonal entangled basis which is not a genuinely entangled basis in $d\otimes d\otimes d$ for $d\geq 3$  \cite{2020Strong}, and they also showed that a strongly nonlocal unextendible product basis(UPB) in $d\otimes d\otimes d$ exists for $d\geq 3$  \cite{shi2021}.  Recently, Wang \emph{et al.} showed a genuinely orthogonal entangled set that is strongly nonlocal in $d\otimes d\otimes d$ for $d\geq 3$ \cite{li2}. The authors in Refs.~\cite{Halder2019Strong,yuan2020strong} also proposed an open question. Whether one can construct strongly nonlocal OPSs in multipartite systems? In this paper, we shall solve this question for any possible three, four and five-partite systems.

	Most of the known constructions of strongly nonlocal OPSs are lack of intuitive structures, and it is not easy to generalize them to multipartite systems.  In this work,  based on the decomposition for the outermost layer of an  $n$-dimensional hypercube for $n=3,4,5$,  we successfully construct  strongly nonlocal OPSs with ``well structure"  of size $d_Ad_Bd_C-(d_A-2)(d_B-2)(d_C-2)$ in $d_A\otimes d_B\otimes d_C$, size $d_Ad_Bd_Cd_D-(d_A-2)(d_B-2)(d_C-2)(d_D-2)$, and size $d_Ad_Bd_Cd_Dd_E-(d_A-2)(d_B-2)(d_C-2)(d_D-2)(d_E-2)$ in $d_A\otimes d_B\otimes d_C\otimes d_D\otimes d_E$ for $d_A,d_B,d_C,d_D,d_E\geq 3$, respectively. Our results answers an open question given in \cite{Halder2019Strong,yuan2020strong} for any possible three, four and five-partite systems.

 \section{Preliminaries}
	In this paper, we only consider pure states and
	positive operator-valued measurement (POVM), and we do not normalize states and operators for simplicity. A local measurement performed to distinguish a set of multipartite orthogonal states is called an \emph{orthogonality-preserving local measurements}, if the postmeasurement states remain orthogonal. An OPS in $d_1\otimes d_2\otimes \cdots\otimes d_n$ is \emph{locally irreducible} if it is not possible to eliminate one or more states from the set by orthogonality-preserving local measurements \cite{Halder2019Strong}. Further, in  $d_1\otimes d_2\otimes\cdots \otimes d_n$, $n\geq 3$ and $d_i\geq 3$, an OPS is \emph{strongly nonlocal} if it is locally irreducible in every bipartition, which shows the phenomenon of strong quantum nonlocality without entanglement \cite{Halder2019Strong}.
	
    An OPS is said to be of  \emph{the strongest nonlocality} if only trivial  orthogonality-preserving POVM can be performed on it for each bipartition of the subsystems \cite{shi2021}. A measurement is \emph{trivial} if all the POVM elements are proportional to the identity operator. By definition, an OPS that is of the strongest nonlocality must be strongly nonlocal. However, the converse is not true usually. For example, a strongly nonlocal OPS in $3\otimes 3\otimes 3$ can be viewed as a strongly nonlocal OPS in $4\otimes 4\otimes 4$. Then Alice can perform a nontrivial orthogonality-preserving POVM $\{\ketbra{3}{3},\bbI-\ketbra{3}{3}\}$ to this OPS.  There exists an efficient way to check whether an OPS is of the strongest nonlocality \cite{shi2021}. Consider an OPS $\{\ket{\psi}\}\subset \otimes_{i=1}^n\cH_{A_i}$. Let $B_1=\{A_2A_3\ldots A_n\}$, $B_2=\{A_3\ldots A_nA_1\}, B_3=\{A_4\ldots A_nA_1A_2\}, \ldots, B_n=\{A_1\ldots A_{n-1}\}$. If the party $B_i$  can only perform a trivial orthogonality-preserving POVM for any $1\leq i\leq n$,  then the OPS  $\{\ket{\psi}\}$ is of the strongest nonlocality. We will use this method throughout the paper.

	For any positive integer $n\geq 2$, we denote $\bbZ_{n}$ as the set $\{0,1,\cdots,n-1\}$, and denote $w_n=e^{\frac{2\pi\sqrt{-1}}{n}}$.	 We  assume that $\{|0\rangle, |1\rangle, \cdots,|n-1\rangle\}$ is the computational basis of the $n$ dimensional Hilbert space $\cH_n$. For any operator $M$ acting on $\cH_n$, we denote the matrix $M$ as the matrix representation of the operator $M$ under the computational basis. In general, we do not distinguish the operator $M$ and its matrix  representation $M$. Given any $n\times n$ matrix $E:=\sum_{i=0}^{n-1}\sum_{j=0}^{n-1} a_{i,j}|i\rangle\langle j|$, we define
	\begin{equation*}
		{}_\mathcal{S}E_{\mathcal{T}}:=\sum_{|s\rangle \in \mathcal{S}}\sum_{|t\rangle \in \mathcal{T}}a_{s,t} |s\rangle\langle t|,
	\end{equation*}
	where $\mathcal{S},\mathcal{T}\subseteq \{|0\rangle,|1\rangle,\cdots, |n-1\rangle\}$. It means that ${}_\mathcal{S}E_{\mathcal{T}}$ is a submatrix of $E$ with  row coordinates $\mathcal{S}$ and column coordinates $\mathcal{T}$.  In the case $\cS=\cT$, we denote
	\begin{equation*}
	  E_{\cS}:={}_{\cS}E_{\cS}
	\end{equation*}
	for simplicity. Finally, we introduce two basic lemmas which are from \cite{shi2021}, and they are useful for showing that an OPS is of the strongest nonlocality. Assume a set $\cS\subseteq\{|0\rangle,|1\rangle,\cdots, |n-1\rangle\}$. An orthogonal set $\{\ket{\psi_i}\}_{i\in\bbZ_s}$ is \emph{spanned} by $\cS$, if for any $i\in\bbZ_s$,  $\ket{\psi_i}$ is a linear combination of the states from $\cS$.

	\begin{lemma}[Block Zeros Lemma \cite{shi2021}]\label{lem:zero}
		Let  an  $n\times n$ matrix $E=(a_{i,j})_{i,j\in\bbZ_n}$ be the matrix representation of an operator  $E=M^{\dagger}M$  under the basis  $\cB:=\{\ket{0},\ket{1},\ldots,\ket{n-1}\}$. Given two nonempty disjoint subsets $\cS$ and $\cT$ of $\cB$, assume  that  $\{\ket{\psi_i}\}_{i\in\bbZ_s}$, $\{\ket{\phi_j}\}_{j\in\bbZ_t}$ are two orthogonal sets  spanned by $\cS$ and $\cT$ respectively, where $s=|\cS|,$ and $t=|\cT|.$  If  $\langle \psi_i| E| \phi_j\rangle =0$
		for any $i\in \mathbb{Z}_s,j\in\mathbb{Z}_t$, then   ${}_\mathcal{S}E_{\mathcal{T}}=\mathbf{0}$  and  ${}_\mathcal{T}E_{\mathcal{S}}=\mathbf{0}$.
	\end{lemma}
	
	\begin{lemma}[Block Trivial  Lemma \cite{shi2021}]\label{lem:trivial}
		Let  an  $n\times n$ matrix $E=(a_{i,j})_{i,j\in\bbZ_n}$ be the matrix representation of an operator  $E=M^{\dagger}M$  under the basis  $\cB:=\{\ket{0},\ket{1},\ldots,\ket{n-1}\}$. Given a nonempty  subset $\cS:=\{\ket{u_0},\ket{u_1},\ldots,\ket{u_{s-1}}\}$  of $\cB$, let $\{\ket{\psi_j} \}_{j=0}^{s-1}$ be an orthogonal  set spanned by $\cS$.     Assume that $\langle \psi_i|E |\psi_j\rangle=0$ for any $i\neq j\in \mathbb{Z}_s$.  If there exists a state $|u_t\rangle \in\cS$,  such that $ {}_{\{|u_t\rangle\}}E_{\cS\setminus \{|u_t\rangle\}}=\mathbf{0}$  and $\langle u_t|\psi_j\rangle \neq 0$  for any $j\in \mathbb{Z}_s$, then  $E_{\cS}\propto \mathbb{I}_{\cS}.$
	\end{lemma}

 \section{$\text{OPSs}$ of the strongest nonlocality}	 	

  In this section, we give some constructions for OPSs of the strongest nonlocality in three, four, and five-partite systems.
	One notes that the OPSs we constructed are all with ``well structure" under the following sense: the states are exactly corresponding to the outermost layer of  an $n$-dimensional hypercube. Through some decompositions for the outermost layer, we can construct the desired OPSs. We use this idea to construct three, four and five-partite OPSs of the strongest nonlocality.

\subsection{OPSs of the strongest nonlocality in three-partite systems}\label{sec:OPS_tri}

		We start the construction OPSs of the strongest nonlocality by an example. Given a 3-dimensional hypercube with coordinates $\{0,1,2\}_A\times \{0,1,2\}_B \times \{0,1,2\}_C$, the outermost layer is $\{0,1,2\}_A\times \{0,1,2\}_B \times \{0,1,2\}_C\setminus \{\{1\}_A\times \{1\}_B \times \{1\}_C\}$. We can decompose the outermost layer in this way: $\cC_1=\{1,2\}_A\times \{0\}_B\times \{0,1\}_C$, $\cC_2=\{1,2\}_A\times \{0,1\}_B\times \{2\}_C$, $\cC_3=\{2\}_A\times \{1,2\}_B\times \{0,1\}_C$, $\cC_4=\{2\}_A\times \{2\}_B\times \{2\}_C$, $\cD_1=\{0,1\}_A\times \{2\}_B\times \{1,2\}_C$, $\cD_2=\{0,1\}_A\times \{1,2\}_B\times \{0\}_C$, $\cD_3=\{0\}_A\times \{0,1\}_B\times \{1,2\}_C$, $\cD_4=\{0\}_A\times \{0\}_B\times \{0\}_C$. See also Fig.~\ref{fig:333cube}. We can obtain an OPS 
	from the decomposition as follows,
	\begin{figure}[h]
		\centering
		\includegraphics[scale=0.55]{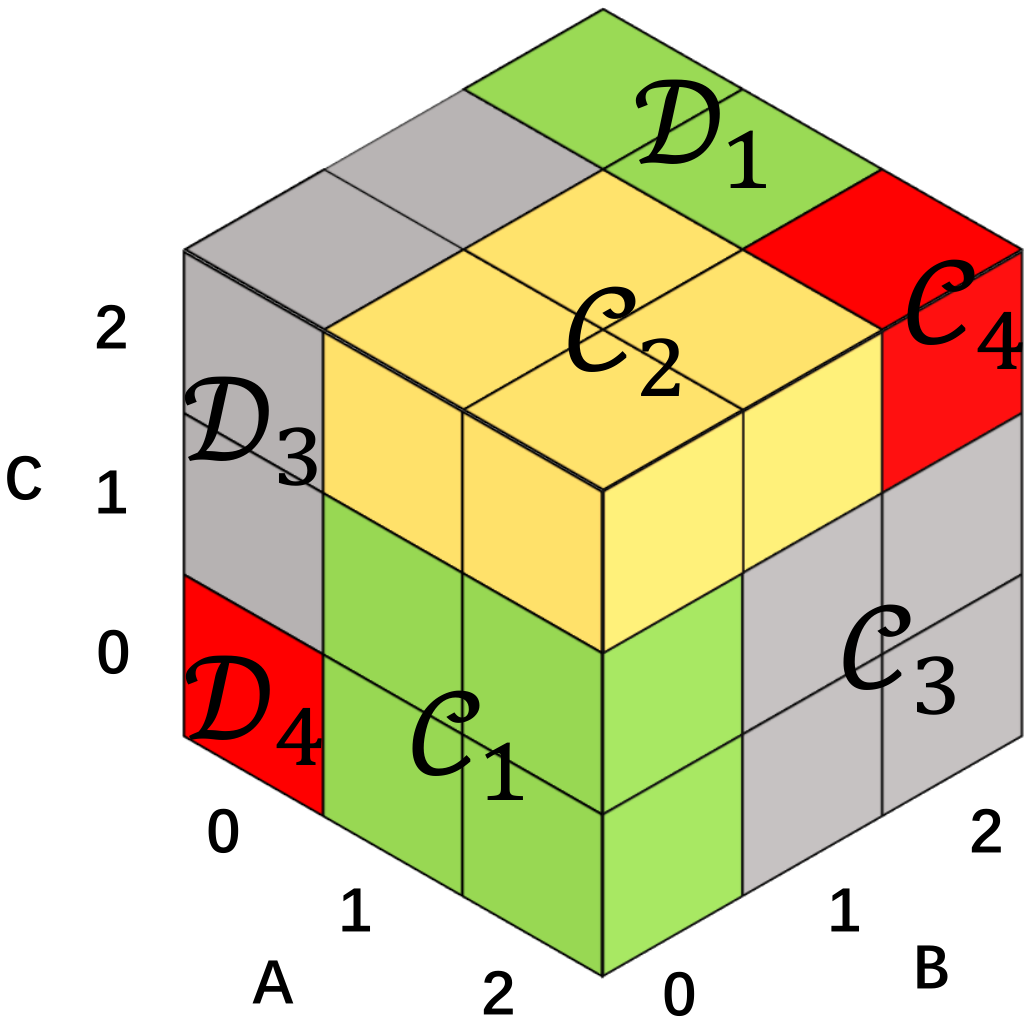}
		\caption{The decomposition for the outermost layer of a 3-dimensional hypercube with coordinates $\{0,1,2\}_A\times \{0,1,2\}_B \times \{0,1,2\}_C$.  }\label{fig:333cube}
	\end{figure}
	\begin{equation}\label{OPB333}
		\begin{aligned}
			\cC_1&:=\{\ket{\xi_i}_A\ket{0}_B\ket{\eta_j}_C\mid (i,j)\in\bbZ_2\times \bbZ_2  \},\\
			\cC_2&:=\{\ket{\xi_i}_A\ket{\eta_j}_B\ket{2}_C\mid(i,j)\in\bbZ_2\times \bbZ_2  \},\\				\cC_3&:=\{\ket{2}_A\ket{\xi_i}_B\ket{\eta_j}_C\mid(i,j)\in\bbZ_2\times \bbZ_2 \},\\
			\cC_4&:= \{ |2\rangle_A|2\rangle_B|2\rangle_C\},\\
			\cD_1&:=\{\ket{\eta_i}_A\ket{2}_B\ket{\xi_j}_C\mid(i,j)\in\bbZ_2\times \bbZ_2 \},\\
			\cD_2&:=\{\ket{\eta_i}_A\ket{\xi_j}_B\ket{0}_C\mid(i,j)\in\bbZ_2\times \bbZ_2 \},\\
			\cD_3&:=\{\ket{0}_A\ket{\eta_i}_B\ket{\xi_j}_C\mid(i,j)\in\bbZ_2\times \bbZ_2 \}, \\
			\cD_4&:= \{ |0\rangle_A|0\rangle_B|0\rangle_C\},\\			
		\end{aligned}
	\end{equation}
	where $|\eta_s\rangle_X=\ket{0}_X+(-1)^s\ket{1}_X, |\xi_s\rangle_X:=\ket{1}_X+(-1)^s\ket{2}_X$ for $s\in \bbZ_2$, 	 for any $X\in\{A,B,C\}$. 		
 Now, we show that $\cup_{i=1}^{4}(\cC_i,\cD_i)$ is of  the strongest nonlocality.

	\begin{example}\label{OPS_SN_333}
		In $3\otimes 3\otimes 3$, the set
		$\cup_{i=1}^{4}(\cC_i,\cD_i)$  given by Eq. \eqref{OPB333} is  an OPS of  the strongest nonlocality. The size of this set is $26$.
	\end{example}

	\begin{proof}
		Let $B$ and $C$ come together to  perform a joint  orthogonality-preserving POVM $\{E=M^{\dagger}M\}$, where $E=(a_{ij,k\ell})_{i,j,k,\ell\in\bbZ_3}$. Then the postmeasurement states $\{\mathbb{I}_A\otimes M\ket{\psi}\big |\ket{\psi}\in \cup_{i=1}^{4}(\cC_i,\cD_i)\}$ should be mutually orthogonal.  That is,
		\begin{equation}\label{OrthogonalOPS333}
		\begin{array}{rcl}
		0&=&{}_A \langle \phi_1| {}_B\langle \phi_2| {}_C\langle \phi_3| \mathbb{I}_A\otimes E |\psi_1\rangle_A|\psi_2\rangle_B|\psi_3\rangle_C = \langle \phi_1|\psi_1\rangle_A ({}_B\langle \phi_2| {}_C\langle \phi_3|  E   |\psi_2\rangle_B|\psi_3\rangle_C)
		\end{array}
		\end{equation}
		whenever $|\phi_1\rangle_A|\phi_2\rangle_B|\phi_3\rangle_C$ and $ |\psi_1\rangle_A|\psi_2\rangle_B|\psi_3\rangle_C$ are two different states from  the set $\cup_{i=1}^{4}(\cC_i,\cD_i)$. Observing that if $\langle \phi_1|\psi_1\rangle_A\neq 0$, one can also get that ${}_B\langle \phi_2| {}_C\langle \phi_3|  E   |\psi_2\rangle_B|\psi_3\rangle_C=0$. Our aim is to show that $E\propto\bbI$ by use this observation.
		
		The eight subsets $\cup_{i=1}^{4}\{\cC_i,\cD_i\}$ in $A|BC$ bipartition correspond to eight blocks of the $3\times 9$ grid in Fig.~\ref{fig:Tri_OPS_333}. For example, $\cC_1$ corresponds to the $2\times 2$ grid $\{(1,2)\times (00,01)\}$. Moreover,  $\cC_i$ is symmetrical to $\cD_i$ $(1\leq i\leq 4)$.

	\begin{figure}[h]
		\centering
		\includegraphics[scale=0.7]{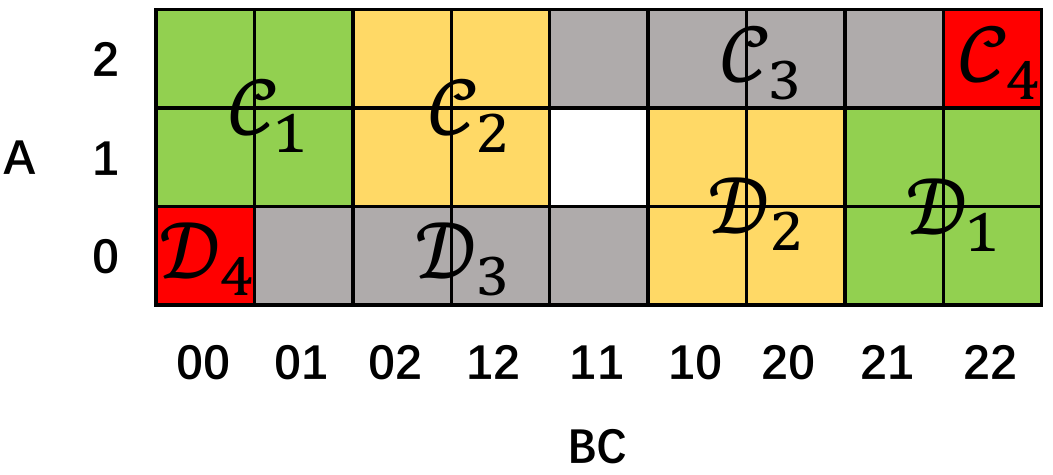}
		\caption{The corresponding  $3\times 9$ grid of $\cup_{i=1}^{4}\{\cC_i,\cD_i\}$ given by Eq. \eqref{OPB333} in $A|BC$ bipartition.  For example, $\cC_1$ corresponds to the $2\times 2$ grid $\{(1,2)\times (00,01)\}$. Moreover,  $\cC_i$ is symmetrical to $\cD_i$ for $1\leq i\leq 4$. }\label{fig:Tri_OPS_333}
	\end{figure}
	
	  We used some notations introduced in \cite{shi2021}. Let $\mathcal{S}=\{\ket{\psi_1}_A\ket{\psi_2}_B\ket{\psi_3}\}$ be a tripartite orthogonal product set. Define
    \begin{equation*}
    \mathcal{S}(|\psi\rangle_A):=\{ |\psi_2\rangle_B|\psi_3\rangle_C   \mid   |\psi\rangle_A |\psi_2\rangle_B|\psi_3\rangle_C \in  \mathcal{S}\}.
   \end{equation*}
   Moreover, define  $\mathcal{S}^{(A)}$ as the support of $\mathcal{S}(|\psi\rangle_A)$  which spans  $\mathcal{S}(|\psi\rangle_A)$.
    For example, in Eq.~\eqref{OPB333}, $\cC_1:=\{\ket{\xi_i}_A\ket{0}_B\ket{\eta_j}_C\mid (i,j)\in\bbZ_2\times \bbZ_2\}$. Then $\cC_1(\ket{\xi_1}_A)=\{\ket{0}_A\ket{\eta_j}_C\}_{j\in \bbZ_2}$, $\cC_1^{(A)}=\{\ket{0}_B\ket{0}_C,\ket{0}_B\ket{1}_C\}$, and $\cC_1(\ket{\xi_1}_A)$ is spanned by $\cC_1^{(A)}$. Actually,  $\{\cC_i^{(A)},\cD_i^{(A)}\}_{i=1}^4$   can be easily observed by Fig.~\ref{fig:Tri_OPS_333}. They are the projection sets of $\{\cC_i,\cD_i\}_{i=1}^4$ in $BC$ party in Fig.~\ref{fig:Tri_OPS_333}.

		Let $\mathcal{V}:=\{|j\rangle_B|k\rangle_C \mid j,k\in \bbZ_3\}$ be the computational basis of the $BC$ party. One finds that $\mathcal{V}$ is a disjoint union of the subsets $\mathcal{C}_1^{(A)},\mathcal{C}_2^{(A)},\mathcal{C}_3^{(A)},\mathcal{C}_4^{(A)}$. That is,
		\begin{equation*}
			\begin{aligned}
				\mathcal{V}&=\mathcal{C}_1^{(A)}\cup\mathcal{C}_2^{(A)}\cup\mathcal{C}_3^{(A)}\cup\mathcal{C}_4^{(A)}, \   \text{ and }  \mathcal{C}_i^{(A)} \cap \mathcal{C}_j^{(A)} =\emptyset
			\end{aligned}
		\end{equation*}
		 whenever $ i\neq j.$

		\noindent{\bf Step 1}   Considering any
		$|\Phi_1\rangle \in 	\mathcal{C}_1(|\xi_0\rangle_A)$,  $|\Phi_2\rangle \in 	\mathcal{C}_2(|\xi_0\rangle_A)$,  $|\Phi_3\rangle \in 	\mathcal{C}_3(|2\rangle_A)$,  $|\Phi_4\rangle \in 	\mathcal{C}_4(|2\rangle_A)$. Since $\ket{\xi_0}_A,\ket{\xi_0}_A,\ket{2}_A,\ket{2}_A$ are mutually non-orthogonal, by Eq.~\eqref{OrthogonalOPS333}, we have
		\begin{equation}
			\langle\Phi_i |E|\Phi_j\rangle=0, \quad \text{for} \ 1\leq i\neq j\leq 4.
		\end{equation}
		Note that  $ \mathcal{C}_1(|\xi_0\rangle_A)$,  $\mathcal{C}_2(|\xi_0\rangle_A)$, $\mathcal{C}_3(|2\rangle_A)$, $\mathcal{C}_4(|2\rangle_A)$  are spanned by $\cC_1^{(A)}$, $\cC_2^{(A)}$, $\cC_3^{(A)}$, $\cC_4^{(A)}$  respectively.  Applying  Lemma \ref{lem:zero}  to any two sets of $ \mathcal{C}_1(|\xi_0\rangle_A)$,  $\mathcal{C}_2(|\xi_0\rangle_A)$, $\mathcal{C}_3(|2\rangle_A)$, $\mathcal{C}_4(|2\rangle_A)$, we obtain
		\begin{equation}
			{}_{\cC_{i}^{(A)}}E_{\cC_{j}^{(A)}}=\mathbf{0}, \quad \text{for} \ 1\leq i\neq j\leq 4.
		\end{equation}
		Therefore, $E$ is a  block  diagonal matrix of the form
		\begin{equation}
			E= E_{\cC_1^{(A)}} \oplus  E_{\cC_2^{(A)}} \oplus  E_{\cC_3^{(A)}}\oplus  E_{\cC_4^{(A)}}.
		\end{equation}
		See also Fig. \ref{Edge333} (I).

				\begin{figure}[h]
		\centering
		\includegraphics[scale=0.42]{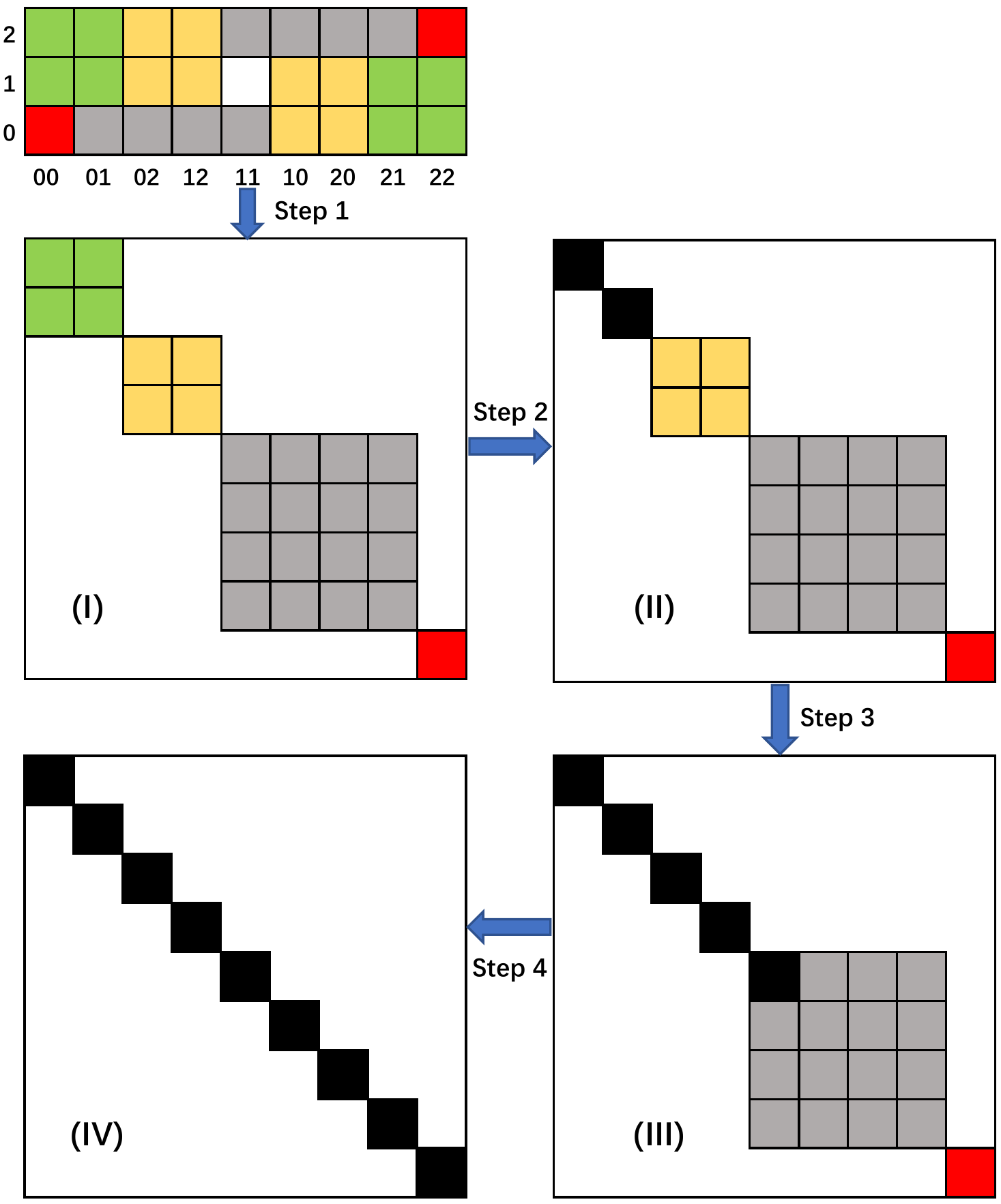}
		\caption{Proving steps of the strongest nonlocality in Example~\ref{OPS_SN_333}.  } \label{Edge333}
	\end{figure}

		\noindent {\bf Step 2} Applying Lemma \ref{lem:zero} to $\mathcal{D}_4(|0\rangle_A)$ and $\cD_3 (|0\rangle_A)$,  we have ${}_{\{|0\rangle_B|0\rangle_C\}}E_{\mathcal{D}_3^{(A)}}=\mathbf{0}$.
		We can also obtain ${}_{\{|0\rangle_B|0\rangle_C\}}E_{\mathcal{C}_1^{(A)}\setminus \{|0\rangle_B|0\rangle_C\}}=\mathbf{0}$ as $\mathcal{C}_1^{(A)}\setminus\{|0\rangle_B|0\rangle_C\}\subset \cD_3^{(A)}$. Applying Lemma \ref{lem:trivial} to the set $\mathcal{C}_1(|\xi_0\rangle_A)$, we obtain
		\begin{equation}
			E_{\mathcal{C}_1^{(A)}}=a\mathbb{I}_{\mathcal{C}_1^{(A)}}.
		\end{equation}  See also Fig. \ref{Edge333} (II).

		\noindent {\bf Step 3}   By Fig. \ref{Edge333} (II), we can obtain the equality  ${}_{\{|0\rangle_B|1\rangle_C\}}E_{\mathcal{D}_3^{(A)}\setminus \{|0\rangle_B|1\rangle_C\}}=\mathbf{0}$.  Applying Lemma \ref{lem:trivial} to the set  $\mathcal{D}_3(|0\rangle_A)$, we have
		$
			E_{\mathcal{D}_3^{(A)}}=b\mathbb{I}_{\mathcal{D}_3^{(A)}}.
		$
		Note that $\cD_3^{(A)}\cap \cC_1^{(A)}\neq \emptyset$. Then $b=a$. Therefore
		\begin{equation}
			E_{\mathcal{C}_1^{(A)}\cup\cD_3^{(A)}}=a\mathbb{I}_{\mathcal{C}_1^{(A)}\cup\cD_3^{(A)}}.
		\end{equation}
		See also Fig. \ref{Edge333} (III).
		
		\noindent {\bf Step 4} By the symmetry of Fig.~\ref{fig:Tri_OPS_333}, we can obtain that $E=a\mathbb{I}$. Thus $E$ is trivial. See also Fig. \ref{Edge333} (IV).
		
		Further, since the eight subsets $\cup_{i=1}^{4}\{\cC_i,\cD_i\}$ in any bipartition of  $\{A|BC, C|AB, B|CA\}$ correspond to a similar grid as Fig.~\ref{fig:Tri_OPS_333},  it implies that $AB$ or $CA$ can only perform a trivial orthogonality-preserving POVM. Thus, the OPS $\cup_{i=1}^{4}(\cC_i,\cD_i)$  given by Eq. \eqref{OPB333} is of  the strongest nonlocality.
	\end{proof}
	\vspace*{0.4cm}

	The above construction and its proof can be straightforwardly extended to the high dimensional systems. We can give a similar  decomposition for the outermost layer of a 3-dimensional hypercube with coordinates $\{0,1,\ldots,d_{A}-1\}_A\times \{0,1,\ldots,d_{B}-1\}_B \times \{0,1,\ldots,d_{C}-1\}_C$. Then we can obtain an OPS
	from the decomposition as follows,
	\begin{equation}\label{OPBd1d2d3}
		\begin{aligned}
			\cC_1&:=\{\ket{\xi_i}_A\ket{0}_B\ket{\eta_j}_C\mid (i,j)\in\bbZ_{d_A-1}\times \bbZ_{d_C-1}  \},\\
			\cC_2&:=\{\ket{\xi_i}_A\ket{\eta_j}_B\ket{d_C-1}_C\mid(i,j)\in\bbZ_{d_A-1}\times \bbZ_{d_B-1}  \},\\
			\cC_3&:=\{\ket{d_A-1}_A\ket{\xi_i}_B\ket{\eta_j}_C\mid(i,j)\in\bbZ_{d_B-1}\times \bbZ_{d_C-1} \},\\
			\mathcal{C}_4&:= \{ |d_A-1\rangle_A|d_B-1\rangle_B|d_C-1\rangle_C\},\\
			\cD_1&:=\{\ket{\eta_i}_A\ket{d_B-1}_B\ket{\xi_j}_C\mid(i,j)\in\bbZ_{d_A-1}\times \bbZ_{d_C-1} \},\\
			\cD_2&:=\{\ket{\eta_i}_A\ket{\xi_j}_B\ket{0}_C\mid(i,j)\in\bbZ_{d_A-1}\times \bbZ_{d_B-1} \},\\
			\cD_3&:=\{\ket{0}_A\ket{\eta_i}_B\ket{\xi_j}_C\mid(i,j)\in\bbZ_{d_B-1}\times \bbZ_{d_C-1} \}, \\
			\mathcal{D}_4&:= \{ |0\rangle_A|0\rangle_B|0\rangle_C\},	
		\end{aligned}
	\end{equation}
	where $\ket{\eta_s}_X=\sum_{t=0}^{d_X-2}w_{d_X-1}^{st}\ket{t}_X$, and $\ket{\xi_s}_X=\sum_{t=0}^{d_X-2}w_{d_X-1}^{st}\ket{t+1}_X$ for $s\in \bbZ_{d_X-1 }$, and $X\in\{A,B,C\}$. Therefore, we have the following theorem.

		\begin{figure}[h]
		\centering
		\includegraphics[scale=0.48]{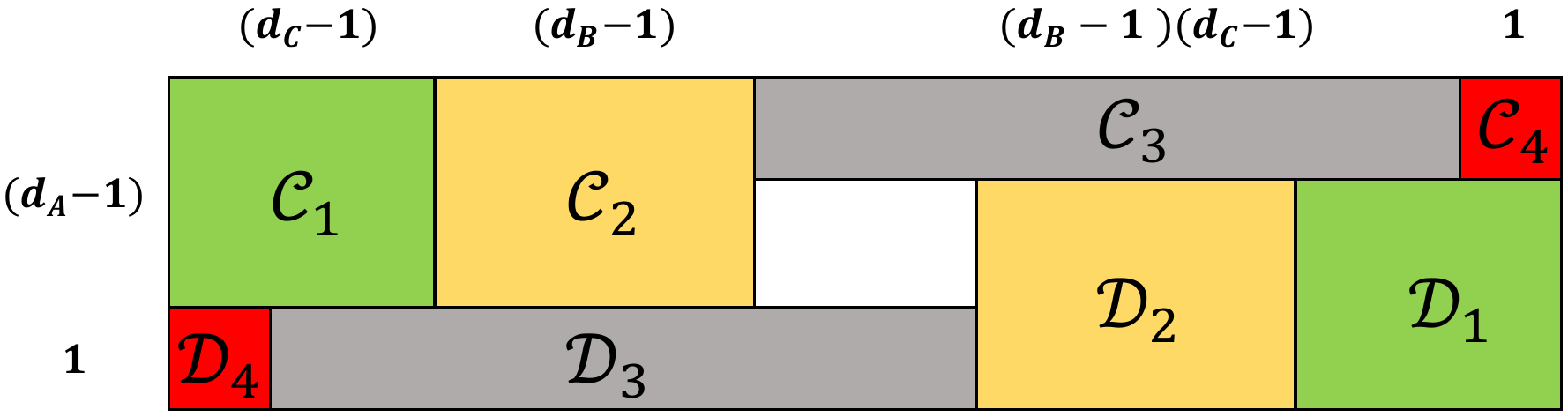}
		\caption{The corresponding  $d_A\times d_Bd_C$ grid of $\cup_{i=1}^{4}\{\cC_i,\cD_i\}$ given by Eq.~\eqref{OPBd1d2d3} in $A|BC$ bipartition.  Moreover,  $\cC_i$ is symmetrical to $\cD_i$ for $1\leq i\leq 4$. }\label{fig:Tri_OPS_dAdBdC}
	\end{figure}
	
	\begin{theorem}\label{OPS_SN_tri}
		In $d_A\otimes d_B\otimes d_C$, $d_A,d_B,d_C\geq 3$, the set $\cup_{i=1}^4(\cC_i,\cD_i)$ given by Eq. \eqref{OPBd1d2d3}  is an  OPS of  the strongest nonlocality. The size of this set is $d_Ad_Bd_C-(d_A-2)(d_B-2)(d_C-2)$.
	\end{theorem}

	\begin{proof}
		The eight subsets $\cup_{i=1}^{4}\{\cC_i,\cD_i\}$ in $A|BC$ bipartition correspond to eight blocks of the $d_A\times d_Bd_C$ grid in Fig.~\ref{fig:Tri_OPS_dAdBdC}. Note that
		\begin{equation}
			\begin{aligned}
				&\cD_{3}^{(A)}\cap\cC_j^{(A)}\neq \emptyset \quad \text{for} \ j=1,2,3;\ \ \  &\cD_{4}^{(A)}\cap\cC_1^{(A)}\neq \emptyset.\\
			\end{aligned}
		\end{equation}
		Since  the eight subsets $\cup_{i=1}^{4}\{\cC_i,\cD_i\}$ in any bipartition of  $\{A|BC, C|AB, B|CA\}$ correspond to a similar grid as Fig.~\ref{fig:Tri_OPS_dAdBdC}, we only need to consider Fig.~\ref{fig:Tri_OPS_dAdBdC}.  Let $B$ and $C$ come together to  perform a joint  orthogonality-preserving POVM $\{E=M^{\dagger}M\}$, where $E=(a_{ij,k\ell})_{i,k\in \bbZ_{d_B}, j,\ell\in \bbZ_{d_C}}$. Then the postmeasurement states $\{\mathbb{I}_A\otimes M\ket{\psi}\big |\ket{\psi}\in \cup_{i=1}^{4}(\cC_i,\cD_i)\}$ should be mutually orthogonal. By Lemma  ~\ref{lem:zero}, similar as the step 1 of  Example~\ref{OPS_SN_333}, via  $\cC_1(\ket{\xi_0}_A), 	\cC_2(\ket{\xi_0}_A), 	\cC_3(\ket{d_A-1}),	\cC_4(\ket{d_A-1})$ we can get
		$$E= E_{\cC_1^{(A)}} \oplus  E_{\cC_2^{(A)}} \oplus  E_{\cC_3^{(A)}}\oplus  E_{\cC_4^{(A)}}.$$ We can complete the  proof  similar as the other steps of Example~\ref{OPS_SN_333}. See Eq.~\eqref{eq:analysis_dAdBdC} for a sketch of the analysis,
		\begin{equation}\label{eq:analysis_dAdBdC}
			\begin{aligned}
				\cD_4(\ket{0}_A),\cD_3(\ket{0}_A)\xrightarrow{\text{Lemma}~\ref{lem:zero}} &  {}_{\cD_4^{(A)}}E_{\cD_3^{(A)}}=\textbf{0},\\
				\cC_1(\ket{\xi_0}_A)\xrightarrow{\text{Lemma}~\ref{lem:trivial}} &  E_{\cC_1^{{(A)}}}=a\bbI_{\cC_1^{{(A)}}},\\
				\cD_3(\ket{0}_A)\xrightarrow{\text{Lemma}~\ref{lem:trivial}} &
				E_{\cD_3^{{(A)}}}=b\bbI_{\cD_3^{{(A)}}}, \\
				\cC_1^{(A)}\cap \cD_3^{(A)}\neq \emptyset\longrightarrow& E_{\mathcal{C}_1^{(A)}\cup\cD_3^{(A)}}=a\mathbb{I}_{\mathcal{C}_1^{(A)}\cup\cD_3^{(A)}}. 		
			\end{aligned}
		\end{equation}
		By the symmetry of Fig.~\ref{fig:Tri_OPS_333}, we can obtain that $E=a\mathbb{I}$. Thus $E$ is trivial.
	\end{proof}
	\vspace{0.4cm}

Very fortunately, we find that the above idea for  constructing strongest nonlocal set of product states can be extended to four and five-partite systems. However, the structure of the states are more complex and difficult than the tripartite cases.  

	\subsection{OPSs of the strongest nonlocality in four-partite systems}\label{sec:OPS_four}

	First, we need to consider the decomposition for the outermost layer of a 4-dimensional hypercube with coordinates $\{0,1,\ldots,d_{A}-1\}_A\times \{0,1,\ldots,d_{B}-1\}_B \times \{0,1,\ldots,d_{C}-1\}_C\times \{0,1,\ldots,d_{D}-1\}_D$. The decomposition can be shown in the following  OPS 
	
	\begin{equation}\label{OPB_four_general}
		\begin{aligned}
			\cC_1&:=\{\ket{\xi_i}_A\ket{\eta_j}_B|0\rangle_C|\xi_k\rangle_D\mid (i,j,k)\in\bbZ_{d_A-1}\times \bbZ_{d_B-1} \times\bbZ_{d_D-1}  \},\\	\cC_2&:=\{\ket{\xi_i}_A\ket{d_B-1}_B|\eta_j\rangle_C|\eta_k\rangle_D\mid      (i,j,k)\in\bbZ_{d_A-1}\times \bbZ_{d_C-1} \times\bbZ_{d_D-1}  \},\\						
			\cC_3&:=\{\ket{\xi_i}_A\ket{\xi_j}_B|\xi_k\rangle_C|d_D-1\rangle_D\mid (i,j,k)\in\bbZ_{d_A-1}\times \bbZ_{d_B-1} \times\bbZ_{d_D-1}  \},\\
			\cC_4&:=\{\ket{\xi_i}_A\ket{d_B-1}_B\ket{0}_C\ket{d_D-1}_D\mid i\in\bbZ_{d_A-1}\},\\
			\cC_5&:=\{\ket{d_A-1}_A\ket{\eta_i}_B|\xi_j\rangle_C|\eta_k\rangle_D\mid (i,j,k)\in\bbZ_{d_B-1}\times \bbZ_{d_C-1} \times\bbZ_{d_D-1}  \},\\	
			\cC_6&:=\{\ket{d_A-1}_A\ket{\eta_i}_B\ket{0}_C\ket{0}_D\mid i\in\bbZ_{d_B-1}\},\\
			\cC_7&:=\{\ket{d_A-1}_A\ket{0}_B\ket{\xi_i}_C\ket{d_D-1}_D\mid i\in\bbZ_{d_C-1}\},\\
			\cC_8&:=\{\ket{d_A-1}_A\ket{d_B-1}_B\ket{d_C-1}_C\ket{\eta_i}_D\mid i\in\bbZ_{d_D-1}\},\\
			\cD_1&:=\{\ket{\eta_i}_A\ket{\xi_j}_B|d_C-1\rangle_C|\eta_k\rangle_D\mid (i,j,k)\in\bbZ_{d_A-1}\times \bbZ_{d_B-1} \times\bbZ_{d_D-1} \},\\
			\cD_2&:=\{\ket{\eta_i}_A\ket{0}_B|\xi_j\rangle_C|\xi_k\rangle_D\mid (i,j,k)\in\bbZ_{d_A-1}\times \bbZ_{d_C-1} \times\bbZ_{d_D-1}  \},\\						
			\cD_3&:=\{\ket{\eta_i}_A\ket{\eta_j}_B|\eta_k\rangle_C|0\rangle_D\mid (i,j,k)\in\bbZ_{d_A-1}\times \bbZ_{d_B-1} \times\bbZ_{d_D-1}   \},\\
			\cD_4&:=\{\ket{\eta_i}_A\ket{0}_B\ket{d_C-1}_C\ket{0}_D\mid i\in\bbZ_{d_A-1}\},\\
			\cD_5&:=\{\ket{0}_A\ket{\xi_i}_B|\eta_j\rangle_C|\xi_k\rangle_D\mid (i,j,k)\in\bbZ_{d_B-1}\times \bbZ_{d_C-1} \times\bbZ_{d_D-1}  \},\\			
			\cD_6&:=\{\ket{0}_A\ket{\xi_i}_B\ket{d_C-1}_C\ket{d_D-1}_D\mid i\in\bbZ_{d_B-1}\},\\
			\cD_7&:=\{\ket{0}_A\ket{d_B-1}_B\ket{\eta_i}_C\ket{0}_D\mid i\in\bbZ_{d_C-1}\},\\
			\cD_8&:=\{\ket{0}_A\ket{0}_B\ket{0}_C\ket{\xi_i}_D\mid i\in\bbZ_{d_D-1}\},
		\end{aligned}
	\end{equation}
    where $\ket{\eta_s}_X=\sum_{t=0}^{d_X-2}w_{d_X-1}^{st}\ket{t}_X$, and $\ket{\xi_s}_X=\sum_{t=0}^{d_X-2}w_{d_X-1}^{st}\ket{t+1}_X$ for $s\in \bbZ_{d_X-1 }$, and $X\in\{A,B,C,D\}$. Now, we show that the above OPS is of the strongest nonlocality.
	
	\begin{theorem}\label{OPS_SN_four_parties_general}
		In $d_A\otimes d_B\otimes d_C\otimes d_D$, $d_A,d_B,d_C,d_D\geq 3$, the set    $\cup_{i=1}^8(\cC_i,\cD_i)$ given by Eq. \eqref{OPB_four_general}  is an  OPS of  the strongest nonlocality. The size of this set is $d_Ad_Bd_Cd_D-(d_A-2)(d_B-2)(d_C-2)(d_D-2)$.
	\end{theorem}
	
	\begin{figure}[h]
		\centering
		\includegraphics[scale=0.5]{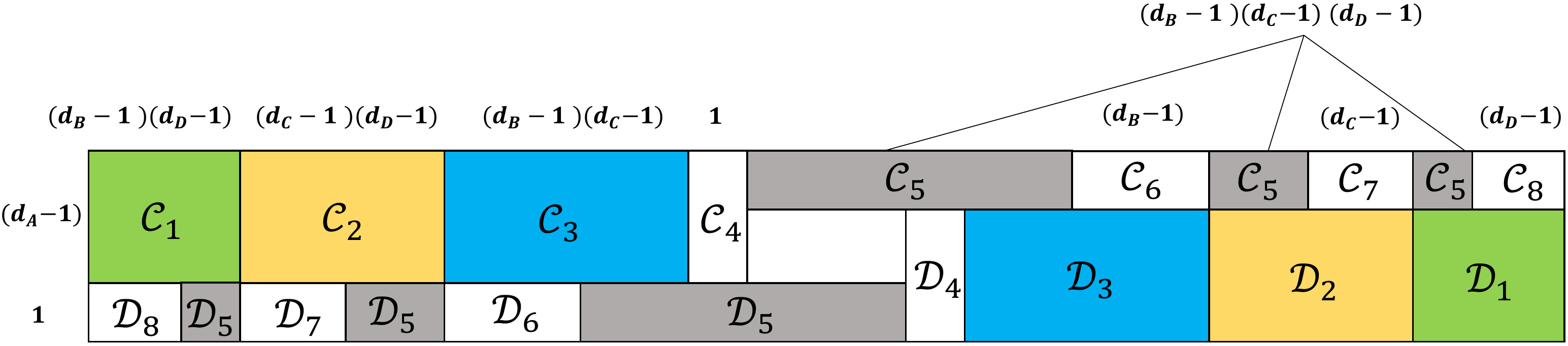}
		\caption{The corresponding  $d_A\times d_Bd_Cd_D$ grid of $\cup_{i=1}^{8}\{\cC_i,\cD_i\}$ given by Eq.~\eqref{OPB_four_general}  in $A|BCD$ bipartition.  Moreover,  $\cC_i$ is symmetrical to $\cD_i$ for $1\leq i\leq 8$. }\label{fig:Four_OPS_dAdBdCdD}
	\end{figure}

	\begin{proof}
		The 16 subsets $\cup_{i=1}^{8}\{\cC_i,\cD_i\}$ in $A|BCD$ bipartition correspond to 16 blocks of the $d_A\times d_Bd_Cd_D$ grid in Fig.~\ref{fig:Four_OPS_dAdBdCdD}. Note that
		\begin{equation}\label{eq:relation_set_dddd}
			\begin{aligned}
				&\cD_{5}^{(A)}\cap\cC_j^{(A)}\neq \emptyset \quad \text{for} \ j=1,2,3,4,5;\\ &\cD_{6}^{(A)}\cap\cC_3^{(A)}\neq \emptyset;\\
				&\cD_{7}^{(A)}\cap\cC_2^{(A)}\neq \emptyset;\\
				&\cD_{8}^{(A)}\cap\cC_1^{(A)}\neq \emptyset.\\
			\end{aligned}
		\end{equation}
		Since  the 16 subsets $\cup_{i=1}^{8}\{\cC_i,\cD_i\}$ in any bipartition of $\{A|BCD,D|ABC, C|DAB, B|CDA\}$ correspond to a similar grid as Fig.~\ref{fig:Four_OPS_dAdBdCdD}, we only need to consider Fig.~\ref{fig:Four_OPS_dAdBdCdD}.  Let $B$, $C$ and $D$ come together to  perform a joint  orthogonality-preserving POVM $\{E=M^{\dagger}M\}$, where $E=(a_{ijk,\ell mn})_{i,\ell\in \bbZ_{d_B}, j,m\in \bbZ_{d_C}, k,n\in \bbZ_{d_D}}$. Then the postmeasurement states $\{\mathbb{I}_A\otimes M\ket{\psi}\big |\ket{\psi}\in \cup_{i=1}^{8}(\cC_i,\cD_i)\}$ should be mutually orthogonal.
		
		\noindent {\bf Step 1}	  Applying Lemma~\ref{lem:zero} to any two elements of $\{\{\cC_i(\ket{\xi_0}_A)\}_{j=1}^4,\{\cC_j(\ket{d_A-1})\}_{j=5}^8\}$,  we have
		\begin{equation}\label{eq:E_8diagonal}
			E=\oplus_{j=1}^8 E_{\cC_j^{(A)}}.
		\end{equation}	
		
		\noindent {\bf Step 2} By Eq.~\eqref{eq:E_8diagonal}, we know that ${}_{\cC_4^{(A)}}E_{\cD_5^{(A)}\setminus\cC_4^{(A)}}=\textbf{0}$. Note that $|\cC_4^{(A)}|=1$, more exactly,  $\cC_4^{(A)}=\{\ket{d_B-1}_B\ket{0}_C\ket{d_D-1}_D\}$. Applying Lemma~\ref{lem:trivial} to $\cD_5({\ket{0}_A})$, we have
		\begin{equation}\label{eq:cD_5trivial}
			E_{\cD_5^{(A)}}=a\bbI_{\cD_5^{(A)}}.
		\end{equation}	
		Next, applying Lemma~\ref{lem:zero} to $\cD_5({\ket{0}_A})$ and $\cD_i({\ket{0}_A})$ for $i=6,7,8$, we obtain
		\begin{equation}\label{eq:D_5D_i}
			{}_{\cD_5^{(A)}}E_{\cD_i^{(A)}}=\textbf{0} \quad \text{for} \ i=6,7,8. \end{equation}	
		
		\noindent {\bf Step 3} Note that $\cD_{6}^{(A)}=\cC_3^{(A)}\setminus \cD_{5}^{(A)},\
		\cD_{7}^{(A)}=\cC_2^{(A)}\setminus \cD_{5}^{(A)},$ and $
		\cD_{8}^{(A)}=\cC_1^{(A)}\setminus \cD_{5}^{(A)}. $
		By Eqs.~\eqref{eq:cD_5trivial} and \eqref{eq:D_5D_i}, we know that for any $\ket{j}_B\ket{k}_C\ket{\ell}_D\in \cD_5^{(A)}\cap\cC_{i}^{(A)}$ for $i=1,2,3$,  ${}_{\{\ket{j}_B\ket{k}_C\ket{\ell}_D\}}E_{\cC_i^{(A)}\setminus\{\ket{j}_B\ket{k}_C\ket{\ell}_D\}}=\textbf{0}$ for $i=1,2,3$.  Applying Lemma~\ref{lem:trivial} to $\cC_i({\ket{\xi_0}_A})$ for $i=1,2,3$. Then we have
		\begin{equation}
			E_{\cC_i^{(A)}}=a_i\bbI_{\cC_i^{(A)}} \quad \text{for} \ i=1,2,3.
		\end{equation}	
		Since $\cD_5^{(A)}\cap\cC_i^{{(A)}}\neq \emptyset$ for $i=1,2,3$, it implies $a_i=a$ for  $i=1,2,3$. Thus 	
		\begin{equation}E_{\{\cup_{i=1}^{3}\cC_i^{(A)}\}\cup \cD_5^{(A)}}=a\bbI_{\{\cup_{i=1}^{3}\cC_i^{(A)}\}\cup \cD_5^{(A)}}.
		\end{equation}	
		
		\noindent {\bf Step 4} By the symmetry of Fig.~\ref{fig:Four_OPS_dAdBdCdD}, we can obtain that $E=a\mathbb{I}$. Thus $E$ is trivial.	
	\end{proof}


	\subsection{OPSs of the strongest nonlocality in five-partite systems}\label{sec:OPS_five}

Consider the cyclic permutation of the eight sets  $\{\ket{0}_A\ket{\xi_i}_B\ket{\eta_j}_C\ket{\xi_k}_D\ket{\eta_\ell}_E\}$, $\{\ket{0}_A\ket{0}_B\ket{\xi_i}_C\ket{d-1}_D\ket{\eta_j}_E\}$,  $\{\ket{0}_A\ket{0}_B\ket{0}_C\ket{\xi_i}_D\ket{\eta_j}_E\}$,  $\{\ket{0}_A\ket{0}_B\ket{0}_C\ket{0}_D\ket{0}_E\}$,  $\{\ket{d-1}_A\ket{\eta_i}_B\ket{\xi_j}_C\ket{\eta_k}_D\ket{\xi_\ell}_E\}$, $\{\ket{d-1}_A\ket{d-1}_B\ket{\eta_i}_C\ket{0}_D\ket{\xi_j}_E\}$,  $\{\ket{d-1}_A\ket{d-1}_B\ket{d-1}_C\ket{\eta_i}_D\ket{\xi_j}_E\}$, $\{\ket{d-1}_A\ket{d-1}_B\ket{d-1}_C\ket{d-1}_D\ket{d-1}_E\}$. Then we can obtain a decomposition with $32$ blocks for the outermost layer of a 5-dimensional hypercube with coordinates $\{0,1,\ldots,d-1\}_A\times \{0,1,\ldots,d-1\}_B \times \{0,1,\ldots,d-1\}_C\times \{0,1,\ldots,d-1\}_D\times \{0,1,\ldots,d-1\}_E$. Similarly, we can obtain a decomposition with $32$ blocks for the outermost layer of a 5-dimensional hypercube with coordinates $\{0,1,\ldots,d_A-1\}_A\times \{0,1,\ldots,d_B-1\}_B \times \{0,1,\ldots,d_C-1\}_C\times \{0,1,\ldots,d_D-1\}_D\times \{0,1,\ldots,d_E-1\}_E$ as follows,

	\begin{equation}\label{OPB_five_general}
		\begin{aligned}
			\cC_1&:=\{\ket{\xi_i}_A\ket{d_B-1}_B\ket{\eta_j}_C\ket{\xi_k}_D\ket{\eta_\ell}_E\mid (i,j,k,\ell)\in\bbZ_{d_A-1}\times \bbZ_{d_C-1} \times\bbZ_{d_D-1} \times\bbZ_{d_E-1} \},\\	\cC_2&:=\{\ket{\xi_i}_A\ket{\eta_j}_B\ket{0}_C\ket{\xi_k}_D\ket{\eta_\ell}_E\mid (i,j,k,\ell)\in\bbZ_{d_A-1}\times \bbZ_{d_B-1} \times\bbZ_{d_D-1} \times\bbZ_{d_E-1} \},\\						
			\cC_3&:=\{\ket{\xi_i}_A\ket{\eta_j}_B\ket{\xi_k}_C\ket{d_D-1}_D\ket{\eta_\ell}_E\mid (i,j,k,\ell)\in\bbZ_{d_A-1}\times \bbZ_{d_B-1} \times\bbZ_{d_C-1} \times\bbZ_{d_E-1} \},\\
			\cC_4&:=\{\ket{\xi_i}_A\ket{\eta_j}_B\ket{\xi_k}_C\ket{\eta_\ell}_D\ket{0}_E\mid (i,j,k,\ell)\in\bbZ_{d_A-1}\times \bbZ_{d_B-1} \times\bbZ_{d_C-1} \times\bbZ_{d_D-1} \},\\
			\cC_5&:=\{\ket{\xi_i}_A\ket{d_B-1}_B\ket{d_C-1}_C\ket{\eta_j}_D\ket{0}_E\mid (i,j)\in\bbZ_{d_A-1}\times \bbZ_{d_D-1} \},\\
			\cC_6&:=\{\ket{\xi_i}_A\ket{d_B-1}_B\ket{\eta_j}_C\ket{0}_D\ket{0}_E\mid (i,j)\in\bbZ_{d_A-1}\times \bbZ_{d_C-1} \},\\
			\cC_7&:=\{\ket{\xi_i}_A\ket{d_B-1}_B\ket{d_C-1}_C\ket{d_D-1}_D\ket{\eta_j}_E\mid (i,j)\in\bbZ_{d_A-1}\times \bbZ_{d_E-1} \},\\
			\cC_8&:=\{\ket{\xi_i}_A\ket{\eta_j}_B\ket{0}_C\ket{0}_D\ket{0}_E\mid (i,j)\in\bbZ_{d_A-1}\times \bbZ_{d_B-1} \},\\
			\cC_9&:=\{\ket{d_A-1}_A\ket{\eta_i}_B\ket{\xi_j}_C\ket{\eta_k}_D\ket{\xi_\ell}_E\mid (i,j,k,\ell)\in \bbZ_{d_B-1} \times\bbZ_{d_C-1} \times\bbZ_{d_D-1}\times \bbZ_{d_E-1}\},\\
			\cC_{10}&:=\{\ket{d_A-1}_A\ket{d_B-1}_B\ket{d_C-1}_C\ket{d_D-1}_D\ket{d_E-1}_E\},\\
			\cC_{11}&:=\{\ket{d_A-1}_A\ket{\eta_i}_B\ket{\xi_j}_C\ket{d_D-1}_D\ket{d_E-1}_E\mid (i,j)\in \bbZ_{d_B-1} \times\bbZ_{d_C-1}\},\\
			\cC_{12}&:=\{\ket{d_A-1}_A\ket{\eta_i}_B\ket{0}_C\ket{0}_D\ket{\xi_j}_E\mid (i,j)\in \bbZ_{d_B-1} \times\bbZ_{d_E-1}\},\\
			\cC_{13}&:=\{\ket{d_A-1}_A\ket{\eta_i}_B\ket{0}_C\ket{\xi_j}_D\ket{d_E-1}_E\mid (i,j)\in \bbZ_{d_B-1} \times\bbZ_{d_D-1}\},\\
			\cC_{14}&:=\{\ket{d_A-1}_A\ket{d_B-1}_B\ket{\eta_i}_C\ket{\xi_j}_D\ket{d_E-1}_E\mid (i,j)\in \bbZ_{d_C-1} \times\bbZ_{d_D-1}\},\\
			\cC_{15}&:=\{\ket{d_A-1}_A\ket{d_B-1}_B\ket{\eta_i}_C\ket{0}_D\ket{\xi_j}_E\mid (i,j)\in \bbZ_{d_C-1} \times\bbZ_{d_E-1}\},\\
			\cC_{16}&:=\{\ket{d_A-1}_A\ket{d_B-1}_B\ket{d_C-1}_C\ket{\eta_i}_D\ket{\xi_j}_E\mid (i,j)\in \bbZ_{d_D-1} \times\bbZ_{d_E-1}\},\\
			\cD_1&:=\{\ket{\eta_i}_A\ket{0}_B\ket{\xi_j}_C\ket{\eta_k}_D\ket{\xi_\ell}_E\mid (i,j,k,\ell)\in\bbZ_{d_A-1}\times \bbZ_{d_C-1} \times\bbZ_{d_D-1} \times\bbZ_{d_E-1} \},\\	\cD_2&:=\{\ket{\eta_i}_A\ket{\xi_j}_B\ket{d_C-1}_C\ket{\eta_k}_D\ket{\xi_\ell}_E\mid (i,j,k,\ell)\in\bbZ_{d_A-1}\times \bbZ_{d_B-1} \times\bbZ_{d_D-1} \times\bbZ_{d_E-1} \},\\						
			\cD_3&:=\{\ket{\eta_i}_A\ket{\xi_j}_B\ket{\eta_k}_C\ket{0}_D\ket{\xi_\ell}_E\mid (i,j,k,\ell)\in\bbZ_{d_A-1}\times \bbZ_{d_B-1} \times\bbZ_{d_C-1} \times\bbZ_{d_E-1} \},\\
			\cD_4&:=\{\ket{\eta_i}_A\ket{\xi_j}_B\ket{\eta_k}_C\ket{\xi_\ell}_D\ket{d_E-1}_E\mid (i,j,k,\ell)\in\bbZ_{d_A-1}\times \bbZ_{d_B-1} \times\bbZ_{d_C-1} \times\bbZ_{d_D-1} \},\\
			\cD_5&:=\{\ket{\eta_i}_A\ket{0}_B\ket{0}_C\ket{\xi_j}_D\ket{d_E-1}_E\mid (i,j)\in\bbZ_{d_A-1}\times \bbZ_{d_D-1} \},\\
			\cD_6&:=\{\ket{\eta_i}_A\ket{0}_B\ket{\xi_j}_C\ket{d_D-1}_D\ket{d_E-1}_E\mid (i,j)\in\bbZ_{d_A-1}\times \bbZ_{d_C-1} \},\\
			\cD_7&:=\{\ket{\eta_i}_A\ket{0}_B\ket{0}_C\ket{0}_D\ket{\xi_j}_E\mid (i,j)\in\bbZ_{d_A-1}\times \bbZ_{d_E-1} \},\\
			\cD_8&:=\{\ket{\eta_i}_A\ket{\xi_j}_B\ket{d_C-1}_C\ket{d_D-1}_D\ket{d_E-1}_E\mid (i,j)\in\bbZ_{d_A-1}\times \bbZ_{d_B-1} \},\\
			\cD_9&:=\{\ket{0}_A\ket{\xi_i}_B\ket{\eta_j}_C\ket{\xi_k}_D\ket{\eta_\ell}_E\mid (i,j,k,\ell)\in \bbZ_{d_B-1} \times\bbZ_{d_C-1} \times\bbZ_{d_D-1}\times \bbZ_{d_E-1}\},\\
			\cD_{10}&:=\{\ket{0}_A\ket{0}_B\ket{0}_C\ket{0}_D\ket{0}_E\},\\
			\cD_{11}&:=\{\ket{0}_A\ket{\xi_i}_B\ket{\eta_j}_C\ket{0}_D\ket{0}_E\mid (i,j)\in \bbZ_{d_B-1} \times\bbZ_{d_C-1}\},\\
			\cD_{12}&:=\{\ket{0}_A\ket{\xi_i}_B\ket{d_C-1}_C\ket{d_D-1}_D\ket{\eta_j}_E\mid (i,j)\in \bbZ_{d_B-1} \times\bbZ_{d_E-1}\},\\
			\cD_{13}&:=\{\ket{0}_A\ket{\xi_i}_B\ket{d_C-1}_C\ket{\eta_j}_D\ket{0}_E\mid (i,j)\in \bbZ_{d_B-1} \times\bbZ_{d_D-1}\},\\
			\cD_{14}&:=\{\ket{0}_A\ket{0}_B\ket{\xi_i}_C\ket{\eta_j}_D\ket{0}_E\mid (i,j)\in \bbZ_{d_C-1} \times\bbZ_{d_D-1}\},\\
			\cD_{15}&:=\{\ket{0}_A\ket{0}_B\ket{\xi_i}_C\ket{d_D-1}_D\ket{\eta_j}_E\mid (i,j)\in \bbZ_{d_C-1} \times\bbZ_{d_E-1}\},\\
			\cD_{16}&:=\{\ket{0}_A\ket{0}_B\ket{0}_C\ket{\xi_i}_D\ket{\eta_j}_E\mid (i,j)\in \bbZ_{d_D-1} \times\bbZ_{d_E-1}\},\\
		\end{aligned}
	\end{equation}
	where $\ket{\eta_s}_X=\sum_{t=0}^{d_X-2}w_{d_X-1}^{st}\ket{t}_X$, and $\ket{\xi_s}_X=\sum_{t=0}^{d_X-2}w_{d_X-1}^{st}\ket{t+1}_X$ for $s\in \bbZ_{d_X-1 }$, and $X\in\{A,B,C,D,E\}$, Then we have the following theorem.

	\begin{theorem}\label{OPS_SN_five_parties_general}
		In $d_A\otimes d_B\otimes d_C\otimes d_D\otimes d_E$, $d_A,d_B,d_C,d_D,d_E\geq 3$, the set $\cup_{i=1}^{16}(\cC_i,\cD_i)$ given by Eq. \eqref{OPB_five_general}  is an OPS of  the strongest nonlocality. The size of this set is $d_Ad_Bd_Cd_Dd_E-(d_A-2)(d_B-2)(d_C-2)(d_D-2)(d_E-2)$.
	\end{theorem}
	
		\begin{figure}[h]
		\centering
		\includegraphics[scale=0.5]{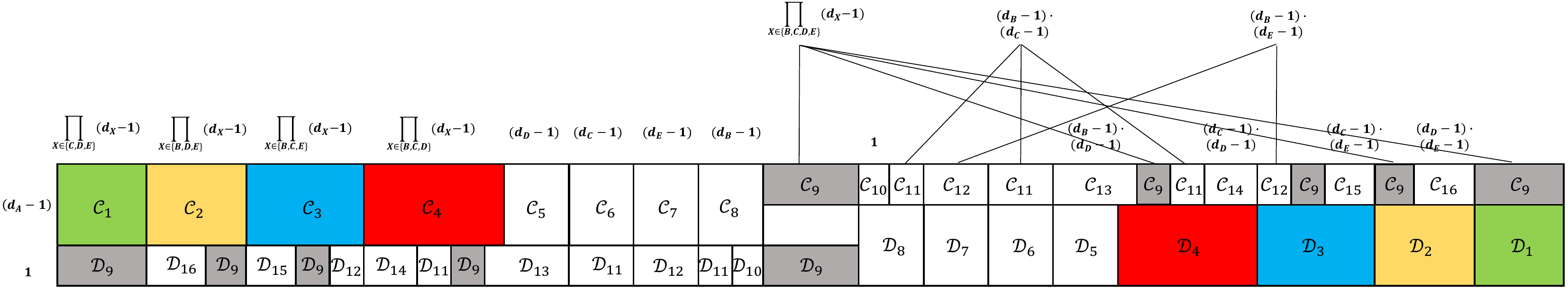}
		\caption{The corresponding  $d_A\times d_Bd_Cd_Dd_E$ grid of $\cup_{i=1}^{16}\{\cC_i,\cD_i\}$ given by Eq.~\eqref{OPB_five_general}  in $A|BCDE$ bipartition.  Moreover,  $\cC_i$ is symmetrical to $\cD_i$ for $1\leq i\leq 16$. }\label{fig:Five_OPS_dAdBdCdDdE}
	\end{figure}

	\begin{proof}
		The 32 subsets $\cup_{i=1}^{16}\{\cC_i,\cD_i\}$ in $A|BCDE$ bipartition correspond to 32 blocks of the $d_A\times d_Bd_Cd_Dd_E$ grid in Fig.~\ref{fig:Five_OPS_dAdBdCdDdE}. Note that
		\begin{equation}\label{eq:relation_set_ddddd}
			\begin{aligned}
				&\cD_{9}^{(A)}\cap\cC_j^{(A)}\neq \emptyset \quad \text{for} \ j=1,2,3,4,9;\\ &\cD_{10}^{(A)}\cap\cC_8^{(A)}\neq \emptyset;\\ &\cD_{11}^{(A)}\cap\cC_j^{(A)}\neq \emptyset \quad \text{for} \ j=4,6,8;\\ &\cD_{12}^{(A)}\cap\cC_j^{(A)}\neq \emptyset \quad \text{for}\ j=3,7;\\
				&\cD_{13}^{(A)}\cap\cC_j^{(A)}\neq \emptyset \quad \text{for}\ j=4,5;\\
				&\cD_{14}^{(A)}\cap\cC_4^{(A)}\neq \emptyset;\\
				&\cD_{15}^{(A)}\cap\cC_3^{(A)}\neq \emptyset;\\
				&\cD_{16}^{(A)}\cap\cC_2^{(A)}\neq \emptyset.
			\end{aligned}
		\end{equation}
		Since  the 32 subsets $\cup_{i=1}^{16}\{\cC_i,\cD_i\}$ in any bipartition of $\{A|BCDE,E|ABCD, D|EABC, C|DEAB, B|CDEA\}$ correspond to a similar grid as Fig.~\ref{fig:Five_OPS_dAdBdCdDdE}, we only need to consider Fig.~\ref{fig:Five_OPS_dAdBdCdDdE}.  Let $B$, $C$, $D$  and $E$ come together to  perform a joint  orthogonality-preserving POVM $\{E=M^{\dagger}M\}$, where $E=(a_{ijk\ell, mnst})_{i,m\in \bbZ_{d_B}, j,n\in \bbZ_{d_C}, k,s\in \bbZ_{d_D},\ell,t\in \bbZ_{d_E}}$. Then the postmeasurement states $\{\mathbb{I}_A\otimes M\ket{\psi}\big |\ket{\psi}\in \cup_{i=1}^{16}(\cC_i,\cD_i)\}$ should be mutually orthogonal.
		
		\noindent {\bf Step 1}	 Applying Lemma~\ref{lem:zero}  to any two elements of $\{\{\cC_j(\ket{\xi_0}_A)\}_{j=1}^8,\{\cC_j(\ket{d_A-1}_A)\}_{j=9}^{16}\}$,  we have
		\begin{equation}\label{eq:diagonal1}
			E= \oplus_{j=1}^{16}E_{\cC_j^{(A)}}.
		\end{equation}	
		
		\noindent {\bf Step 2}  Applying Lemma~\ref{lem:zero} to any two elements of $\{\cD_i(\ket{0}_A)\}_{i=9}^{16}$, we have
		\begin{equation}\label{eq:cD_icD_j}
			{}_{\cD_{i_1}^{(A)}}E_{\cD_{i_2}^{(A)}}=\textbf{0} \quad \text{for} \ 9\leq i_1\neq i_2\leq 16.
		\end{equation}	
		
		Next, since $\cC_{8}^{(A)}\setminus\cD_{10}^{(A)}\subset \cD_{11}^{(A)}$, ${}_{\cD_{10}^{(A)}}E_{\cC_{8}^{(A)}\setminus\cD_{10}^{(A)}}=\textbf{0}$. Note that $|\cD_{10}^{(A)}|=1$, more exactly, $\cD_{10}^{(A)} =\{\ket{0}_B\ket{0}_C\ket{0}_D\ket{0}_E\}$.  Applying Lemma~\ref{lem:trivial} to $\cC_8({\ket{\xi_0}_A})$, we obtain
		\begin{equation}\label{eq:cC_8trivial}
			E_{\cC_8^{(A)}}=a\bbI_{\cC_8^{(A)}}.
		\end{equation}

		\noindent {\bf Step 3}
		By Eqs.~\eqref{eq:diagonal1} and \eqref{eq:cC_8trivial}, we know that for any $\ket{j}_B\ket{k}_C\ket{\ell}_D\ket{m}_E\in \cD_{11}^{(A)}\cap\cC_{8}^{(A)}$, ${}_{\{\ket{j}_B\ket{k}_C\ket{\ell}_D\ket{m}_E\}}E_{\cD_{11}^{(A)}\setminus\{\ket{j}_B\ket{k}_C\ket{\ell}_D\ket{m}_E\}}=\textbf{0}$.
		Applying Lemma~\ref{lem:trivial} to $\cD_{11}({\ket{0}_A})$, we have
		\begin{equation}\label{eq:cC_11trivial}
			E_{\cD_{11}^{(A)}}=a\bbI_{\cD_{11}^{(A)}}.
		\end{equation}
		
		Next, by Eqs.~\eqref{eq:cD_icD_j} and \eqref{eq:cC_11trivial}, we know that for any $\ket{j}_B\ket{k}_C\ket{\ell}_D\ket{m}_E\in \cD_{11}^{(A)}\cap\cC_{4}^{(A)}$,
		${}_{\{\ket{j}_B\ket{k}_C\ket{\ell}_D\ket{m}_E\}}E_{\cC_{4}^{(A)}\setminus\{\ket{j}_B\ket{k}_C\ket{\ell}_D\ket{m}_E\}}=\textbf{0}$.
		Applying Lemma~\ref{lem:trivial} to $\cC_{4}({\ket{\xi_0}_A})$, we have
		\begin{equation}\label{eq:cC_4trivial}
			E_{\cC_{4}^{(A)}}=a\bbI_{\cC_{4}^{(A)}}.
		\end{equation}
		
		\noindent {\bf Step 4}
		By Eqs.~\eqref{eq:diagonal1} and \eqref{eq:cC_4trivial}, we know that for any $\ket{j}_B\ket{k}_C\ket{\ell}_D\ket{m}_E\in \cD_{i}^{(A)}\cap\cC_{4}^{(A)}$ for $i=9,13$,
		${}_{\{\ket{j}_B\ket{k}_C\ket{\ell}_D\ket{m}_E\}}E_{\cD_{i}^{(A)}\setminus\{\ket{j}_B\ket{k}_C\ket{\ell}_D\ket{m}_E\}}=\textbf{0}$.
		Applying Lemma~\ref{lem:trivial} to $\cD_{i}({\ket{0}_A})$ for $i=9,13$, we have
		\begin{equation}\label{eq:cC_9,13trivial}
			E_{\cD_{i}^{(A)}}=a\bbI_{\cD_{i}^{(A)}} \quad \text{for} \ i=9,13.
		\end{equation}
		
		Next, By Eqs.~\eqref{eq:cD_icD_j} and \eqref{eq:cC_9,13trivial}, we know that for any $\ket{j}_B\ket{k}_C\ket{\ell}_D\ket{m}_E\in \cD_{9}^{(A)}\cap\cC_{i}^{(A)}$ for $i=2,3$,
		${}_{\{\ket{j}_B\ket{k}_C\ket{\ell}_D\ket{m}_E\}}E_{\cC_{i}^{(A)}\setminus\{\ket{j}_B\ket{k}_C\ket{\ell}_D\ket{m}_E\}}=\textbf{0}$.
		Applying Lemma~\ref{lem:trivial} to $\cC_{i}({\ket{\xi_0}_A})$ for $i=2,3$, we have
		\begin{equation}\label{eq:cC_2,3trivial}
			E_{\cC_{i}^{(A)}}=a\bbI_{\cC_{i}^{(A)}} \quad \text{for} \ i=2,3.
		\end{equation}
		
		\noindent {\bf Step 5}
		By Eqs.~\eqref{eq:diagonal1} and \eqref{eq:cC_2,3trivial}, we know that for any $\ket{j}_B\ket{k}_C\ket{\ell}_D\ket{m}_E\in \cD_{12}^{(A)}\cap\cC_{3}^{(A)}$,
		${}_{\{\ket{j}_B\ket{k}_C\ket{\ell}_D\ket{m}_E\}}E_{\cD_{12}^{(A)}\setminus\{\ket{j}_B\ket{k}_C\ket{\ell}_D\ket{m}_E\}}=\textbf{0}$.
		Applying Lemma~\ref{lem:trivial} to $\cD_{12}({\ket{0}_A})$, we have
		\begin{equation}\label{eq:cC_1trivial}
			E_{\cD_{12}^{(A)}}=a\bbI_{\cD_{12}^{(A)}}.
		\end{equation}
		
		Thus, above all,
		\begin{equation}
			E_{\{\cup_{i=1}^{8}\cC_i^{(A)}\}\cup\cD_9^{(A)}}=\bbI_{\{\cup_{i=1}^{8}\cC_i^{(A)}\}\cup\cD_9^{(A)}}
		\end{equation}
		as 	$\{\cup_{i=1}^{8}\cC_i^{(A)}\}\cup\cD_9^{(A)}$ is equal to $ \cC_8^{(A)} \cup \cD_{11}^{(A)} \cup  \cC_{4}^{(A)} \cup \cD_{9}^{(A)} \cup \cD_{13}^{(A)} \cup \cC_{2}^{(A)} \cup\cC_{3}^{(A)} \cup\cD_{12}^{(A)}.$
		
		\noindent {\bf Step 6} By the symmetry of Fig.~\ref{fig:Five_OPS_dAdBdCdDdE}, we can obtain that $E=a\mathbb{I}$. Thus $E$ is trivial.	
	\end{proof}
	\vskip 8pt
	
	Any of our OPS in three, four and five-partite systems can be extended to a complete orthogonal product basis (COPB). We only need to add the product states which lie in the inside layer of the hypercube. Thus, our results solve an open question asked by the authors in Refs.~\cite{Halder2019Strong, yuan2020strong}. That is, strongly nonlocal OPSs and COPBs do exist in all possible three, four and five-partite systems.  However, we are unable to generalize this structure to any $n$-partite system for $n>5$. There are several difficulties. First, we cannot obtain a general expression of the OPS from the outermost layer
of an $n$-dimensional hypercube,  from which the OPS has a similiar structure in the bipartitions $\{A_1|A_2A_3\cdots A_n, A_2|A_3\cdots A_n A_1,  \ldots,  A_n|A_1A_2\cdots A_{n-1} \}$.  Second, all of our proofs for the strongest nonlocality of the OPSs are based on the grid representations of the OPSs in  the biparition $A_1|A_2A_3\cdots A_n$, like Figs.~\ref{fig:Tri_OPS_dAdBdC}, \ref{fig:Four_OPS_dAdBdCdD} and \ref{fig:Five_OPS_dAdBdCdDdE}, and we do not know the grid representations in the biparition $A_1|A_2A_3\cdots A_n$ when $n> 5$. Third, we may take more steps to arrive at the statement $E\propto \mathbb{I}$. Therefore, it's more difficult to constructing OPSs of strong nonlocality for quantum system with more parties. However, we believe that our method can be generalized to $n$-partite systems for $n>5$.

\section{Conclusion and Discussion}
 We have constructed strongly nonlocal OPSs in three, four and five-partite systems.  This result has answered an open question in Refs.~\cite{Halder2019Strong,yuan2020strong} for any possible three, four and five-partite systems. There are some interesting problems left.  How to generalize our constructions to $n$-partite systems for $n>5$? Can we construct UPBs from our structures?
	
	One can use entanglement as a resource to finish the  locally distinguished protocol  when the given set is locally indistinguishable. Such protocol is called the entanglement-assisted discrimination~\cite{ghosh2001distinguishability,cohen2008understanding,bandyopadhyay2016entanglement,zhang2016entanglement,gungor2016entanglement,zhang2018local,Sumit2019Genuinely,zhang2020locally,Shi2020Unextendible,2020Strong}.   It is also to quantify how many entanglement resource are needed to local distinguish the strongly nonlocal set constructed here in all bipartitions.
	
\section*{Acknowledgments}
\label{sec:ack}	
	FS and XZ were supported by NSFC under Grant No. 11771419,  the Fundamental Research Funds for the Central Universities, and Anhui Initiative in Quantum Information Technologies under Grant No. AHY150200.  MSL and MHY were supported  by  National  Natural  Science  Foundation  of  China  (12005092, 11875160,   and  U1801661), the China Postdoctoral Science Foundation (2020M681996), the  Natural  Science  Foundation  of  Guang-dong  Province  (2017B030308003),   the  Key  R$\&$D  Program of   Guangdong   province   (2018B030326001),   the   Guang-dong    Innovative    and    Entrepreneurial    Research    TeamProgram (2016ZT06D348), the Science, Technology and   Innovation   Commission   of   Shenzhen   Municipality (JCYJ20170412152620376   and   JCYJ20170817105046702 and  KYTDPT20181011104202253),   the Economy, Trade  and  Information  Commission  of  Shenzhen Municipality (201901161512).   YLW is supported by   the NSFC under Grant No. 11901084 and the Research startup funds of DGUT under Grant No. GC300501-103.  LC and MH were supported by the  NNSF of China (Grant No. 11871089), and the Fundamental Research Funds for the Central Universities (Grant No. ZG216S2005).


	







	\bibliographystyle{IEEEtran}
	\bibliography{reference}	

\begin{thebibliography}{10}
\providecommand{\url}[1]{#1}
\csname url@samestyle\endcsname
\providecommand{\newblock}{\relax}
\providecommand{\bibinfo}[2]{#2}
\providecommand{\BIBentrySTDinterwordspacing}{\spaceskip=0pt\relax}
\providecommand{\BIBentryALTinterwordstretchfactor}{4}
\providecommand{\BIBentryALTinterwordspacing}{\spaceskip=\fontdimen2\font plus
\BIBentryALTinterwordstretchfactor\fontdimen3\font minus
  \fontdimen4\font\relax}
\providecommand{\BIBforeignlanguage}[2]{{%
\expandafter\ifx\csname l@#1\endcsname\relax
\typeout{** WARNING: IEEEtran.bst: No hyphenation pattern has been}%
\typeout{** loaded for the language `#1'. Using the pattern for}%
\typeout{** the default language instead.}%
\else
\language=\csname l@#1\endcsname
\fi
#2}}
\providecommand{\BIBdecl}{\relax}
\BIBdecl

\bibitem{terhal2001hiding}
B.~M. Terhal, D.~P. Divincenzo, and D.~Leung, ``Hiding bits in bell states.''
  \emph{Phys. Rev. Lett.}, vol.~86, no.~25, pp. 5807--5810, 2001.

\bibitem{divincenzo2002quantum}
D.~P. Divincenzo, D.~Leung, and B.~M. Terhal, ``Quantum data hiding,''
  \emph{IEEE Trans. Inf. Theory}, vol.~48, no.~3, pp. 580--598, 2002.

\bibitem{eggeling2002hiding}
T.~Eggeling and R.~F. Werner, ``Hiding classical data in multipartite quantum
  states,'' \emph{Phys. Rev. Lett.}, vol.~89, no.~9, p. 097905, 2002.

\bibitem{Matthews2009Distinguishability}
W.~Matthews, S.~Wehner, and A.~Winter, ``Distinguishability of quantum states
  under restricted families of measurements with an application to quantum data
  hiding,'' \emph{Commun. Math. Phys.}, vol. 291, no.~3, pp. p.813--843, 2009.

\bibitem{Markham2008Graph}
D.~Markham and B.~C. Sanders, ``Graph states for quantum secret sharing,''
  \emph{Phys. Rev. A}, vol.~78, no.~4, pp. 144--144, 2008.

\bibitem{Hillery}
M.~Hillery, V.~Buzek, and A.~Berthiaume, ``Quantum secret sharing,''
  \emph{Phys. Rev. A}, vol.~59, p. 1829, 1999.

\bibitem{Rahaman}
R.~Rahaman and M.~G.~Parker, ``Quantum scheme for secret sharing based on local
  distinguishability,'' \emph{Phys. Rev. A}, vol.~91, p. 022330, 2015.

\bibitem{bennett1999quantum}
C.~H. Bennett, D.~P. Divincenzo, C.~A. Fuchs, T.~Mor, E.~M. Rains, P.~W. Shor,
  J.~A. Smolin, and W.~K. Wootters, ``Quantum nonlocality without
  entanglement,'' \emph{Phys. Rev. A}, vol.~59, no.~2, pp. 1070--1091, 1999.

\bibitem{1}
D.~P. Divincenzo, T.~Mor, P.~W. Shor, J.~A. Smolin, and B.~M. Terhal,
  ``Unextendible product bases, uncompletable product bases and bound
  entanglement,'' \emph{Commun. Math. Phys.}, vol. 238, no.~3, pp. 379--410,
  2003.

\bibitem{2}
Y.~Feng and Y.~Shi, ``Characterizing locally indistinguishable orthogonal
  product states,'' \emph{IEEE Trans. Inf. Theory}, vol.~55, no.~6, pp.
  p.2799--2806, 2009.

\bibitem{3}
J.~Niset and N.~J. Cerf, ``Multipartite nonlocality without entanglement in
  many dimensions,'' \emph{Phys. Rev. A}, vol.~74, p. 052103, 2006.

\bibitem{4}
Y.~Yang, F.~Gao, G.~Tian, T.~Cao, and Q.~Wen, ``Local distinguishability of
  orthogonal quantum states in a 2$\otimes$2$\otimes$2 system,'' \emph{Phys.
  Rev. A}, vol.~88, no.~2, p. 024301, 2013.

\bibitem{5}
S.~Halder, ``Several nonlocal sets of multipartite pure orthogonal product
  states,'' \emph{Phys. Rev. A}, vol.~98, p. 022303, 2018.

\bibitem{6}
G.~Xu, Q.~Wen, F.~Gao, S.~Qin, and H.~Zuo, ``Local indistinguishability of
  multipartite orthogonal product bases,'' \emph{Quantum Inf. Process.},
  vol.~16, no.~11, p. 276, 2017.

\bibitem{7}
Y.-L. Wang, M.-S. Li, Z.-J. Zheng, and S.-M. Fei, ``The local
  indistinguishability of multipartite product states,'' \emph{Quantum Inf.
  Process.}, vol.~16, no.~1, pp. 1--13, 2017.

\bibitem{8}
Z.-C. Zhang, K.-J. Zhang, F.~Gao, Q.-Y. Wen, and C.~H. Oh, ``Construction of
  nonlocal multipartite quantum states,'' \emph{Phys. Rev. A}, vol.~95, p.
  052344, 2017.

\bibitem{9}
S.~Ghosh, G.~Kar, A.~Roy, and D.~Sarkar, ``Distinguishability of maximally
  entangled states,'' \emph{Phys. Rev. A}, vol.~70, p. 022304, 2004.

\bibitem{10}
H.~Fan, ``Distinguishability and indistinguishability by local operations and
  classical communication,'' \emph{Phys. Rev. Lett.}, vol.~92, p. 177905, 2004.

\bibitem{11}
M.~Nathanson, ``Distinguishing bipartitite orthogonal states using locc: Best
  and worst cases,'' \emph{J. Math. Phys}, vol.~46, no.~6, p. 062103, 2005.

\bibitem{12}
N.~Yu, R.~Duan, and M.~Ying, ``Any $2\otimes n$ subspace is locally
  distinguishable,'' \emph{Phys. Rev. A}, vol.~84, p. 012304, 2011.

\bibitem{13}
R.~Duan, Y.~Feng, Z.~Ji, and M.~Ying, ``Distinguishing arbitrary multipartite
  basis unambiguously using local operations and classical communication,''
  \emph{Phys. Rev. Lett.}, vol.~98, p. 230502, 2007.

\bibitem{14}
S.~Bandyopadhyay, S.~Ghosh, and G.~Kar, ``Locc distinguishability of
  unilaterally transformable quantum states,'' \emph{New J. Phys.}, vol.~13,
  no.~12, p. 123013, 2011.

\bibitem{15}
A.~Cosentino, ``Positive-partial-transpose-indistinguishable states via
  semidefinite programming,'' \emph{Phys. Rev. A}, vol.~87, no.~1, 2013.

\bibitem{16}
N.~Yu, R.~Duan, and M.~Ying, ``Four locally indistinguishable ququad-ququad
  orthogonal maximally entangled states,'' \emph{Phys. Rev. Lett.}, vol. 109,
  no.~2, p. 020506, 2012.

\bibitem{17}
S.~Bandyopadhyay, ``Entanglement, mixedness, and perfect local discrimination
  of orthogonal quantum states,'' \emph{Phys. Rev. A}, vol.~85, p. 042319,
  2012.

\bibitem{Halder2019Strong}
S.~Halder, M.~Banik, S.~Agrawal, and S.~Bandyopadhyay, ``Strong quantum
  nonlocality without entanglement,'' \emph{Phys. Rev. Lett.}, vol. 122, no.~4,
  p. 040403, 2019.

\bibitem{zhangstrong2019}
Z.-C. Zhang and X.~Zhang, ``Strong quantum nonlocality in multipartite quantum
  systems,'' \emph{Phys. Rev. A}, vol.~99, p. 062108, 2019.

\bibitem{yuan2020strong}
P.~Yuan, G.~Tian, and X.~Sun, ``Strong quantum nonlocality without entanglement
  in multipartite quantum systems,'' \emph{Phys. Rev. A}, vol. 102, p. 042228,
  2020.

\bibitem{2020Strong}
F.~Shi, M.~Hu, L.~Chen, and X.~Zhang, ``Strong quantum nonlocality with
  entanglement,'' \emph{Phys. Rev. A}, vol. 102, p. 042202, 2020.

\bibitem{shi2021}
F.~Shi, M.-S. Li, M.~Hu, L.~Chen, M.-H. Yung, Y.-L. Wang, and X.~Zhang,
  ``Strongly nonlocal unextendible product bases do exist,''
  \emph{arXiv:2101.00735v2}, 2021.

\bibitem{li2}
Y.-L. Wang, M.-S. Li, and M.-H. Yung, ``Graph-connectivity-based strong quantum
  nonlocality with genuine entanglement,'' \emph{Phys. Rev. A}, vol. 104, p.
  012424, 2021.

\bibitem{ghosh2001distinguishability}
S.~Ghosh, G.~Kar, A.~Roy, A.~Sen, U.~Sen \emph{et~al.}, ``Distinguishability of
  bell states,'' \emph{Phys. Rev. Lett.}, vol.~87, no.~27, p. 277902, 2001.

\bibitem{cohen2008understanding}
S.~M. Cohen, ``Understanding entanglement as resource: Locally distinguishing
  unextendible product bases,'' \emph{Phys. Rev. A}, vol.~77, no.~1, p. 012304,
  2008.

\bibitem{bandyopadhyay2016entanglement}
S.~Bandyopadhyay, S.~Halder, and M.~Nathanson, ``Entanglement as a resource for
  local state discrimination in multipartite systems,'' \emph{Phys. Rev. A},
  vol.~94, no.~2, p. 022311, 2016.

\bibitem{zhang2016entanglement}
Z.-C. Zhang, F.~Gao, T.-Q. Cao, S.-J. Qin, and Q.-Y. Wen, ``Entanglement as a
  resource to distinguish orthogonal product states,'' \emph{Sci. Rep.},
  vol.~6, no.~1, pp. 1--7, 2016.

\bibitem{gungor2016entanglement}
{\"O}.~G{\"u}ng{\"o}r and S.~Turgut, ``Entanglement-assisted state
  discrimination and entanglement preservation,'' \emph{Phys. Rev. A}, vol.~94,
  no.~3, p. 032330, 2016.

\bibitem{zhang2018local}
Z.-C. Zhang, Y.-Q. Song, T.-T. Song, F.~Gao, S.-J. Qin, and Q.-Y. Wen, ``Local
  distinguishability of orthogonal quantum states with multiple copies of
  2$\otimes$2 maximally entangled states,'' \emph{Phys. Rev. A}, vol.~97,
  no.~2, p. 022334, 2018.

\bibitem{Sumit2019Genuinely}
S.~Rout, A.~G. Maity, A.~Mukherjee, S.~Halder, and M.~Banik, ``Genuinely
  nonlocal product bases: Classification and entanglement-assisted
  discrimination,'' \emph{Phys. Rev. A}, vol. 100, no.~3, p. 032321, 2019.

\bibitem{zhang2020locally}
Z.-C. Zhang, X.~Wu, and X.~Zhang, ``Locally distinguishing unextendible product
  bases by using entanglement efficiently,'' \emph{Phys. Rev. A}, vol. 101,
  no.~2, p. 022306, 2020.

\bibitem{Shi2020Unextendible}
F.~Shi, X.~Zhang, and L.~Chen, ``Unextendible product bases from tile
  structures and their local entanglement-assisted distinguishability,''
  \emph{Phys. Rev. A}, vol. 101, no.~6, p. 062329, 2020.

\end{thebibliography}
\end{document}